\documentclass[11pt]{llncs}
\usepackage{latexsym,amssymb,amsmath,amstext,amsfonts,owna4}
\usepackage{enumitem,algorithm,algcompatible}

\spnewtheorem{postulate}{Postulate}{\bfseries}{\rmfamily}
\spnewtheorem{fact}{Fact}{\bfseries}{\itshape}

\ifthenelse{\pdfoutput = 0}
{ \message{We are running LaTeX.} 
  \usepackage[dvips]{graphicx}
  \DeclareGraphicsExtensions{.eps}
  \usepackage[hypertex,bookmarks=false]{hyperref}}
{ \message{We are running PDFLaTeX.}
  \usepackage[pdftex]{color}
  \usepackage[pdftex]{graphicx}
  \DeclareGraphicsExtensions{.pdf}}

\title{Behavioural Theory of Reflective Algorithms I: Reflective Sequential Algorithms\thanks{The research reported in this paper has been partially supported by the {\bf Austrian Science Fund (FWF: [I2420-N31])} for the project \emph{Higher-Order Logics and Structure}.}\thanks{Under review for publication in Theoretical Computer Science.}}

\author{Klaus-Dieter Schewe\inst{1}, Flavio Ferrarotti\inst{2}}

\institute{Zhejiang University, UIUC Institute, Haining, China, \email{kd.schewe@intl.zju.edu.cn, kdschewe@acm.org}
\and
Software Competence Center Hagenberg, agenberg, Austria, \email{flavio.ferrarotti@gmail.com}}

\begin{document}

\maketitle

\begin{abstract}
We develop a behavioural theory of reflective sequential algorithms (RSAs), i.e. sequential algorithms that can modify their own behaviour. The theory comprises a set of language-independent postulates defining the class of RSAs, an abstract machine model, and the proof that all RSAs are captured by this machine model. As in Gurevich's behavioural theory for sequential algorithms RSAs are sequential-time, bounded parallel algorithms, where the bound depends on the algorithm only and not on the input. Different from the class of sequential algorithms every state of an RSA includes a representation of the algorithm in that state, thus enabling linguistic reflection. Bounded exploration is preserved using terms as values. The model of reflective sequential abstract state machines (rsASMs) extends sequential ASMs using extended states that include an updatable representation of the main ASM rule to be executed by the machine in that state. Updates to the representation of ASM signatures and rules are realised by means of a sophisticated tree algebra.


\end{abstract}

\setcounter{footnote}{0}

\section{Introduction}

Adaptivity (or self-adaptation) refers to the ability of a system to change its own behaviour. In the context of programming this concept, known under the term {\em linguistic reflection}, appears already in LISP~\cite{smith:popl1984}, where programs and data are both represented uniformly as lists, and thus programs represented as data can be executed dynamically by means of an evaluation operator. Run-time and compile-time linguistic reflection in programming and database research have been investigated in general by Stemple, Van den Bussche and others in \cite{stemple:2000,bussche:jcss1996}. Recently, adaptivity has attracted again a lot of interest in research, in particular in connection with systems of (cyber-physical) systems \cite{riccobene:abz2014}.

However, this raises the question of the theoretical foundations of adaptive systems. In this article we make a first step in this direction by means of a behavioural theory of reflective sequential algorithms. A {\em behavioural theory} in general comprises an axiomatic definition of a class of algorithms or systems by means of a set of characterising postulates, and an abstract machine model together with the proof that the abstract machine model captures the given class of algorithms or systems. The proof comprises two parts, one showing that every instance of the abstract machine model satisfies the postulates ({\em plausibility}), and another one showing that all algorithms stipulated by the postulates can be step-by-step simulated by an abstract machine model instance ({\em characterisation}).

\subsection{Our Contribution}

In this article we investigate a behavioural theory for reflective sequential algorithms (RSAs) extending and cleansing our previous sketch in \cite{ferrarotti:psi2017}\footnote{Actually, there is even a gap in the main proof in \cite{ferrarotti:psi2017}, as the final statements concerning the ``self representation'' of a sequential algorithms are not at all obvious, but require some sophisticated arguments.}. Our contributions are threefold: 

\begin{enumerate}

\item We provide an axiomatic, language-independent definition of RSAs.

\item We define an extension of sequential ASMs to reflective sequential ASMs (rsASMs), by means of which RSAs can be specified.

\item We prove that RSAs are captured by rsASMs, i.e. rsASMs satisfy the postulates of our axiomatisation, and any RSA as stipulated by the axiomatisation can be defined by a behaviourally equivalent rsASM\footnote{The notion of {\em behavioural equivalence} will be slightly weaker than the corresponding notion for sequential algorithms, as there is no need to require that changes to the represented algorithm are exactly the same. Instead it suffices to postulate that the represented algorithm will produce essentially the same updates disregarding those that are irrelevant.}.

\end{enumerate}

Concerning the axiomatisation an RSA should proceed in sequential time, i.e. we should have states, initial states and a state transition function same as for sequential algorithms. Concerning the notion of state we argue on grounds of the behavioural theory of sequential algorithms that we can always assume a finite representation of a sequential algorithm, no matter how abstract this representation is. Consequently, states of an RSA can be defined by universal algebras (aka Tarski structures). However, the difference to the abstract state postulate for sequential algorithms is that the signature over which states are formed comprises a standard subsignature used to represent the state that is manipulated by the algorithm and a subsignature used for the representation of the algorithm itself. Furthermore, the update set in a state is defined by the algorithm represented in the state. 

While the background operations that are needed to access the represented algorithm can be left abstract it is nonetheless essential to postulate that terms over the standard subsignature are used in this representation. This implies that these terms have a dual nature: they can be interpreted in every state to give values of a base set as required in a Tarski structure, but they are also used as values themselves thus defining extended base sets. If a term is treated as a value, then its interpretation is the term itself, so we need a function \textit{raise\/} that turns a term regarded as value in an extended base set into an interpretable term. As common in linguistic reflection this function has an inverse called \textit{drop\/}. The existence of \textit{raise\/} and \textit{drop\/} and the possibility to create new function symbols in the active signature taking them from a \textit{reserve\/} constitute reasonable requirements for the background structure of an RSA. 

The most tricky aspect of the axiomatisation is the generalisation of bounded exploration. Naturally, as the represented algorithm may change in each step and this may include increasing the standard subsignature, a fixed set of ground terms that determines update sets in every state as in the bounded exploration postulate for sequential algorithms cannot exist. However, as all means of an algorithm to change itself must appear somehow in the algorithm's description, we argue that there is still a bounded exploration witness, i.e. a set of ground terms that determines the update sets yielded in a state. However, a double interpretation will be required: the first interpretation may result in terms over the standard subsignature, which then can again be interpreted to define the values needed in the updates. Phrased differently, the first interpretation can be seen as resulting in a bounded exploration witness for the represented algorithm.

Concerning our second contribution, the definition of rsASMs, this is much more straightforward, as only a concrete representation of a sequential ASM is required\footnote{In this article we choose a self-representation by means of trees similar to syntax trees exploited in parsing theory and compiler construction. Nonetheless, there are many other suitable ways for such a self representation, which would also lead to an abstract machine model satisfying the postulates. For the behavioural theory only the existence of an abstract machine model capturing the class of reflective sequential algorithms is important, while for applications that exploit the machine model for concrete specifications it is left open, which self representation will be more handy.}. In each step the rule is taken from a dedicated location \textit{self\/}, which uses a tree structure to represent the signature and rule, and a sophisticated tree algebra to manipulate tree values \cite{schewe:jucs2010}. We also exploit partial updates in the form of \cite{schewe:ejc2011} to minimise clashes that may otherwise result from simultaneously updating \textit{self\/} by several parallel branches.

Concerning the plausibility and characterisation proofs the former one requires a rather straightforward construction of a bounded exploration witness from an rsASM, for which the representation using  \textit{self\/} is essential, while the latter one will be accomplished by a sequence of lemmata, the key problem being that there is a theoretically unbounded number of different algorithms that nonetheless have to be handled uniformly.

\subsection{Organisation of the Article}

The remainder of this article is organised as follows. In Section \ref{sec:asm} we summarise preliminaries, first giving a brief summary of sequential ASMs \cite{boerger:2003}. Then we introduce a tree algebra and partial updates, by means of which bulk tree data structures can be handled in ASMs. In Section \ref{sec:rasm} we introduce rsASMs, which are based on sequential ASMs with the differences discussed above and a background structure capturing tree structures and tree algebra operations. We use rsASMs to provide examples for reflective sequential algorithms. Section \ref{sec:rsa} is then dedicated to the first part of the behavioural theory, i.e. the axiomatic definition of RSAs. As discussed above the key problems concern the self-representation and the extension of bounded exploration. We also show that rsASMs satisfy our postulates, thus they define RSAs. In Section \ref{sec:theory} we approach the more difficult part of the theory, i.e. the proof that every RSA as stipulated by the postulates can be modelled by a behaviourally equivalent rsASM. In Section \ref{sec:others} we embed our work into other behavioural theories, and we conclude with a summary and outlook in Section \ref{sec:schluss}, in particular with respect to extending the theory to cover unbounded parallelism and concurrency as sketched in \cite{schewe:acsw2017}.

\section{Preliminaries}\label{sec:asm}\label{sec:trees}

In this section we first recall some basic definitions about sequential ASMs. Then we present tree structures and the tree algebra adopted from \cite{schewe:jucs2010} to manipulate trees. These will be exploited in the next section for the self representation of sequential algorithms. We further extend the rules of ASMs slightly using partial updates as in \cite{schewe:ejc2011}, which are particularly useful when dealing with bulk structures such as trees.

\subsection{Sequential ASMs}

Fix a signature $\Sigma$, i.e. a finite set of function symbols, such that each $f \in \Sigma$ has an arity $ar(f)$. A {\em state} $S$ over $\Sigma$ is given by a base set $B$ and an interpretation of the function symbols in $\Sigma$ by functions $f_S : B^n \rightarrow B$ for $n=ar(f)$. Using $\Sigma$ we can define {\em (ground) terms} in the usual way. Then we obtain an evaluation function, which defines for each term $t \in \mathbb{T}$ its value $\text{val}_S(t)$ in a state $S$. A sequential ASM {\em rule} over $\Sigma$ is defined as follows:

\begin{description}

\item[assignment.] Whenever $t_i$ ($i=0,\dots,n$) are terms over $\Sigma$ and $f \in \Sigma$ has arity $n$, then $f(t_1,\dots,t_n) := t_0$ is a rule.

\item[branching.] If $r_+$ and $r_-$ are rules and $\varphi$ is a Boolean term, then also \texttt{IF} $\varphi$ \texttt{THEN} $r_+$ \texttt{ELSE} $r_-$ \texttt{ENDIF} is a rule.

\item[bounded parallel composition.] If $r_1, \dots, r_k$ are rules, then also \texttt{PAR} $r_1 \dots r_k$ \texttt{ENDPAR} is a rule.

\item[let.] If $r$ is a rule, $x$ a variable and $t$ a term, then \texttt{LET} $x = t$ \texttt{IN} $r$ is a rule.

\end{description}

The let rules are just a syntactic construct that eases specifications, but technically they can be omitted. In particular, the proof in \cite{gurevich:tocl2000} shows that let rules can be dispensed with.

Each rule can be interpreted in a state, and doing so yields an update set. In general, a {\em location} is a pair $\ell = (f,(a_1,\dots,a_n))$ with a function symbol $f \in \Sigma$ of arity $n$ and an $n$-tuple of values from the base set $B$. An {\em update} is a pair $(\ell,a_0)$ with a value $a_0 \in B$. In a state $S$ the update set $\Delta_r(S)$ defined by a rule $r$ is yielded as follows:

\begin{itemize}

\item If $r$ is an assignment rule $f(t_1,\dots,t_n) := t_0$, then
\[ \Delta_r(S) = \{ ((f,(\text{val}_S(t_1) ,\dots, \text{val}_S(t_n))), \text{val}_S(t_0)) \} . \]

\item If $r$ is a branching rule \texttt{IF} $\varphi$ \texttt{THEN} $r_+$ \texttt{ELSE} $r_-$ \texttt{ENDIF}, then
\[ \Delta_r(S) = \begin{cases} 
\Delta_{r_+}(S) &\text{if}\; \text{val}_S(\varphi) = \textbf{true} \\
\Delta_{r_-}(S) &\text{if}\; \text{val}_S(\varphi) = \textbf{false} \end{cases}
\]

\item If $r$ is a parallel rule \texttt{PAR} $r_1 \dots r_k$ \texttt{ENDPAR}, then $\Delta_r(S) = \bigcup_{i=1}^k \Delta_{r_i}(S)$.

\item If $r$ is a let rule \texttt{LET} $x = t$ \texttt{IN} $r^\prime$, then substituting $t$ for $x$ in $r^\prime$ defines $\Delta_r(S) = \Delta_{\{ x \mapsto t \}.r^\prime}(S)$.

\end{itemize}

An update set $\Delta$ is {\em consistent} iff it does not contain clashes, i.e. whenever $(\ell,v_1), (\ell,v_2) \in \Delta$ hold, then we must have $v_1 = v_2$. If $\Delta_r(S)$ is consistent, it defines a {\em successor state} $S^\prime = S + \Delta_r(S)$ by
\[ val_{S^\prime}(\ell) = \begin{cases} v_0 &\text{for}\; (\ell,v_0) \in \Delta_r(S) \\
val_S(\ell) &\text{else} \end{cases} \]

In accordance with \cite{gurevich:tocl2000} we extend this definition by $S + \Delta_r(S) = S$ in case $\Delta_r(S)$ is inconsistent.

\begin{definition}\rm

A {\em sequential ASM} comprises a signature $\Sigma$ defining the set $\mathcal{S}$ of states, a set $\mathcal{I} \subseteq \mathcal{S}$ of initial states over $\Sigma$, a sequential ASM rule $r$ over $\Sigma$ and a transition function $\tau$ on states over $\Sigma$ with $\tau(S) = S + \Delta_r(S)$ for all states $S$. Both $\mathcal{S}$ and $\mathcal{I}$ are closed under isomorphisms.

\end{definition}

This definition of successor states gives rise to the notion of a run. A {\em run} of a sequential ASM is a sequence $S_0, S_1, \dots$ of states with $S_0 \in \mathcal{I}$ and each $S_{i+1}$ being a successor of $S_i$.

\subsection{The Background}

For sequential ASMs we usually make some implicit assumptions without further mentioning them. However, as we will see later, for our envisioned extensions these assumptions will gain more importance. First, not all values of $B$ are used in a state $S$. We therefore assume a set {\em reserve} containing the unused values of the base set. In order to support partial functions we further assume a value {\em undef} in every base set. To avoid conflicts with non-strict functions we further assume that $f_S (v_1, \dots, v_n) = \textit{undef\/}$ holds, whenever one of the arguments $v_i$ is {\em undef}. Furthermore, truth values \textbf{true} and \textbf{false} are also assumed to occur in $B$.

Concerning the terms we assume that logical operators $\wedge$, $\neg$, $\vee$, etc. are defined. These as well as as equality can be used in the same way as the function symbols in $\Sigma$ in the definition of terms. We further assume that values in the base set define ground terms that are interpreted by themselves. In this way we extend the set $\mathbb{T}$ of (ground) terms, and the evaluation function $\text{val}_S$ in a state $S$.

All these assumptions about {\em reserve}, {\em undef}, truth values, operators on them and equality define the minimum requirements for the {\em background}. There may be further constants and operators, which are in general defined by a background class. According to \cite{blass:beatcs2007} a background class is determined by a background signature consisting of constructor and function symbols, the latter ones associated with a fixed arity, while for constructor symbols it is also permitted that the arity is unfixed or bounded.

\begin{definition}\label{def-bg-class}\rm

Let $\mathcal{D}$ be a set of base domains and $V_K$ a background signature, then a {\em background class} $\mathcal{K}$ with signature $V_K$ over $\mathcal{D}$ comprises a universe $U$ and an interpretation of function symbols in $V_K$ over $U$. The universe is defined as $U =
\bigcup \mathcal{D}^\prime$, where $\mathcal{D}^\prime$ is the smallest set with $\mathcal{D} \subseteq \mathcal{D}^\prime$ satisfying the following properties for each constructor symbol
$\llcorner\lrcorner\in V_K$:

\begin{itemize}

\item If $\llcorner\lrcorner\in V_K$ has unfixed arity, then $\llcorner D \lrcorner \in \mathcal{D}^\prime$ holds for all $D \in \mathcal{D}^\prime$, and $\llcorner a_1,\dots,a_m \lrcorner \in {\llcorner D \lrcorner}$ for every $m \in \mathbb{N}$ and $a_1,\dots,a_m \in D$, and $A_{\llcorner \lrcorner} \in \mathcal{D}^\prime$ with $A_{\llcorner \lrcorner}=\bigcup_{\llcorner  D \lrcorner\in \mathcal{D}^\prime
} \llcorner D \lrcorner$.

\item If $\llcorner\lrcorner\in V_K$ has bounded arity $n$, then $\llcorner D_1 ,\dots, D_m \lrcorner \in \mathcal{D}^\prime$ for all $m \le n$ and $D_i \in \mathcal{D}^\prime$ ($1 \le i \le m$), and $\llcorner a_1,\dots,a_m \lrcorner \in {\llcorner D_1 ,\dots, D_m \lrcorner}$ for every $m \in \mathbb{N}$ and $a_1,\dots,a_m \in D$.

\item If $\llcorner\lrcorner\in V_K$ has fixed arity $n$, then $\llcorner D_1 ,\dots, D_n \lrcorner \in \mathcal{D}^\prime$ for all $D_i \in \mathcal{D}^\prime$ and $\llcorner a_1,\dots,a_n \lrcorner \in {\llcorner D_1 ,\dots, D_n \lrcorner}$ for all $a_i \in D_i$ ($1 \le i \le n$),

\end{itemize}

\end{definition}

Thus, a sequential ASM must contain an infinite set \textit{reserve\/} of reserve values, truth values and their connectives, the equality predicate, the undefinedness value {\em undef}, and a background class $\mathcal{K}$ defined by a background signature $V_K$. While the requirements for the background of a sequential ASM are minimal and therefore usually left implicit, the background will become important for reflective sequential ASMs.

\subsection{Tree Structures}

We now provide the details of the tree structures and the tree algebra. These structures will be defined as part of the background of an rsASM in the next section.

\begin{definition} \label{def-unranked-tree}

An {\em unranked tree structure} is a structure $(\mathcal{O},\prec_{c},\prec_{s})$ consisting of a finite, non-empty set $\mathcal{O}$ of node identifiers, called tree domain, and irreflexive relations $\prec_{c}$ (child relation) and $\prec_{s}$ (sibling relation) over $\mathcal{O}$ satisfying the following conditions:

\begin{itemize}

\item there exists a unique, distinguished node $o_r \in \mathcal{O}$ (root) such that for all $o \in \mathcal{O} - \{ o_r \}$ there is exactly one $o^\prime \in \mathcal{O}$ with $o^\prime \prec_c o$, and

\item whenever $o_1 \prec_s o_2$ holds, then there is some $o \in \mathcal{O}$ with $o \prec_c o_i$ for $i=1,2$.

\end{itemize}

\end{definition}

For $x_1 \prec_{c} x_2$ we say that $x_2$ is a {\em child} of $x_1$. For $x_1 \prec_{s} x_2$ we say that $x_2$ is the {\em next sibling} of $x_1$, and $x_1$ is the {\em previous sibling} of $x_2$. In order to obtain trees from this, we add labels and assign values to the leaves. For this we fix a finite, non-empty set $L$ of {\em labels}, and a finite family $\{ \tau_i \}_{i \in i}$ of data types. Each data type $\tau_i$ is associated with a {\em value domain} $dom(\tau_i)$. The corresponding {\em universe} $U$ contains all possible values of these data types, i.e. $U=\bigcup_{i \in I}dom(\tau_i)$.

\begin{definition}\label{def-tree}

A {\em tree} $t$ over the set of labels $L$ with values in the universe $U$ comprises an unranked tree structure $\gamma_t=(\mathcal{O}_t,\prec_{c},\prec_{s})$, a total label function $\omega_t: \mathcal{O}_t \rightarrow L$, and a partial value function $\upsilon_t: \mathcal{O}_t \rightarrow U$ that is defined on the leaves in $\gamma_t$.

\end{definition}

Let $T_L$ denote the set of all trees with labels in $L$, and let $root(t)$ denote the root node of a tree $t$. A sequence $t_1,...,t_k$ of trees is called a {\em hedge}, and a multiset $\langle t_1,...,t_k \rangle$ of trees is called a {\em forest}. Let $\epsilon$ denote the empty hedge, and let $H_L$ denote the set of all hedges with labels in $L$. A tree $t_1$ is a {\em subtree} of $t_2$ (notation $t_1 \sqsubseteq t_2$) iff the following properties are satisfied:

\begin{enumerate}

\item $\mathcal{O}_{t_1} \subseteq \mathcal{O}_{t_2}$, 

\item $o_1 \prec_c o_2$ holds in $t_1$ iff it holds in $t_2$, 

\item $o_1 \prec_s o_2$ holds in $t_1$ iff it holds in $t_2$, 

\item $\omega_{t_1}(o^\prime) = \omega_{t_2}(o^\prime)$ holds for all $o^\prime \in \mathcal{O}_{t_1}$, and 

\item for all leaves $o^\prime \in \mathcal{O}_{t_1}$ we have $\upsilon_{t_1}(o^\prime) = \upsilon_{t_2}(o^\prime)$. 

\end{enumerate}

$t_1$ is the {\em largest subtree} of $t_2$ (denoted as $\widehat{o}$) at node $o$ iff $t_1 \sqsubseteq t_2$ with $root(t_1) = o$ and there is no tree $t_3$ with $t_1 \neq t_3 \neq t_2$ such that $t_1 \sqsubseteq t_3 \sqsubseteq t_2$.  

\begin{definition}\label{def-context}

The {\em set of contexts} $C_L$ over $L$ is the set $T_{L \cup \{\xi\}}$ of trees with labels in $L \cup \{\xi\}$ ($\xi \notin L$) such that for each tree $t \in C_L$ exactly one leaf node is labelled with $\xi$ and the value assigned to this leaf is \textit{undef\/}.

\end{definition}

The context with a single node labelled $\xi$ is called trivial and is simply denoted as $\xi$. Contexts allow us to define substitution operations that replace a subtree of a tree or context by a new tree or context. This leads to the following four kinds of substitutions:

\begin{description}

\item[Tree-to-tree substitution.] For a tree $t_1 \in T_{L_1}$, a node $o \in \mathcal{O}_{t_1}$ and a tree $t_2 \in T_{L_2}$ the result $\textit{subst}_{tt}(t_1, o, t_2) = t_1[\widehat{o} \mapsto t_2]$ of substituting $t_2$ for the subtree rooted at $o$ is a tree in $T_{L_1 \cup L_2}$.

\item[Tree-to-context substitution.] For a tree $t_1 \in T_{L_1}$, a node $o \in \mathcal{O}_{t_1}$ the result $\textit{subst}_{tc}(t_1, o, \xi) = t_1[\widehat{o} \mapsto \xi]$ of substituting the trivial context for the subtree rooted at $o$ is a context in $C_{L_1}$.

\item[Context-to-context substitution.] For contexts $c_1 \in C_{L_1}$ and $c_2 \in C_{L_2}$ the result $\textit{subst}_{cc}(c_1, c_2) = c_1[\xi \mapsto c_2]$ of substituting $c_2$ for the leaf labelled by $\xi$ in $c_1$ is a context in $C_{L_1 \cup L_2}$.

\item[Context-to-tree substitution.] For a context $c_1 \in C_{L_1}$ and a tree $t_2 \in T_{L_2}$ the result $\textit{subst}_{ct}(c_1, t_2) = c_1[\xi \mapsto t_2]$ of substituting $t_2$ for the leaf labelled by $\xi$ in $c_1$ is a tree in $T_{L_1 \cup L_2}$.

\end{description}

As a shortcut we also write $\textit{subst}_{tc}(t_1, o, c_2)$ for $\textit{subst}_{cc}(\textit{subst}_{tc}(t_1, o, \xi), c_2)$, which is a context in $C_{L_1 \cup L_2}$.

\subsection{Tree Algebra}

To provide manipulation operations over trees at a level higher than individual nodes and edges, we need constructs to select arbitrary tree portions. For this we provide two selector constructs, which result in subtrees and contexts, respectively. For a tree $t = (\gamma_t, \omega_t, \upsilon_t) \in T_L$ these constructs are defined as follows:

\begin{itemize}

\item $context: \mathcal{O}_t \times \mathcal{O}_t \rightarrow C_L$ is a partial function on pairs $(o_1,o_2)$ of nodes with $o_1 \prec_c^+ o_2$ and $context(o_1,o_2) = \textit{subst\/}_{tc}(\widehat{o}_1, o_2, \xi) = \widehat{o}_1[{\widehat{o}_2\mapsto\xi}]$, where $\prec_c^+$ denotes the transitive closure of $\prec_c$.

\item $subtree: \mathcal{O}_t \rightarrow T_L$ is a function defined by $subtree(o) = \widehat{o}$.

\end{itemize}

The {\em set $\mathbb{T}$ of tree algebra terms} over $L \cup \{ \epsilon, \xi \}$ comprises label terms, hedge terms, and context terms, i.e. $\mathbb{T} = L \cup \mathbb{T}_h \cup \mathbb{T}_c$, which are defined as follows:

\begin{itemize}

\item The set $\mathbb{T}_h$ is the smallest set with $T_L \subseteq \mathbb{T}_h$ such that (1) $\epsilon \in \mathbb{T}_h$, (2) $a \langle h \rangle \in \mathbb{T}_h$ for $a \in L$ and $h \in
\mathbb{T}_h$, and (3) $t_1 \dots t_n \in \mathbb{T}_h$ for $t_i \in \mathbb{T}_h$ ($i=1 ,\dots, n$). 

\item The set of context terms $\mathbb{T}_c$ is the smallest set with (1) $\xi \in \mathbb{T}_c$
and (2) $a \langle t_1 ,\dots, t_n \rangle \in \mathbb{T}_c$ for $a \in L$ and terms $t_1 ,\dots, t_n \in \mathbb{T}_h \cup \mathbb{T}_c$, such that exactly one $t_i$ is a context term in $\mathbb{T}_c$.

\end{itemize}

With these we now define the operators of our tree algebra as follows:

\begin{description}

\item[label\_hedge.] The operator \textit{label\_hedge\/} turns a hedge into a tree with a new added root, i.e.
\[ \textit{label\_hedge\/}(a, t_1 \dots t_n) = a \langle t_1, \dots, t_n \rangle \]

\item[label\_context.] Similarly, the operator \textit{label\_context\/} turns a context into a context with a new added root, i.e.
\[ \textit{label\_context\/}(a, c) = a \langle c \rangle \]

\item[left\_extend.] The operator \textit{left\_extend\/} integrates the trees in a hedge into a context extending it on the left, i.e. 
\[ \textit{left\_extend\/}(t_1 \dots t_n, a \langle t_1^\prime ,\dots, t_m^\prime \rangle) = a \langle t_1, \dots, t_n, t_1^\prime, \dots, t_m^\prime \rangle \]

\item[right\_extend.] Likewise, the operator \textit{right\_extend\/} integrates the trees in a hedge into a context extending it on the right, i.e. 
\[ \textit{right\_extend\/}(t_1 \dots t_n, a \langle t_1^\prime ,\dots, t_m^\prime \rangle) = a \langle t_1^\prime, \dots, t_m^\prime, t_1, \dots, t_n \rangle \]

\item[concat.] The operator \textit{concat\/} simply concatenates two hedges, i.e.  
\[ \textit{concat\/}(t_1 \dots t_n, t_1^\prime \dots t_m^\prime) = t_1 \dots t_n t_1^\prime \dots t_m^\prime \]

\item[inject\_hedge.] The operator \textit{inject\_hedge\/} turns a context into a tree by substituting a hedge for $\xi$, i.e. 
\[ \textit{inject\_hedge\/}(c, t_1 \dots t_n) = c[\xi \mapsto t_1 \dots t_n] \]

\item[inject\_context.] The operator \textit{inject\_context\/} substitutes a context for $\xi$, i.e. 
\[ \textit{inject\_context\/}(c_1,c_2) = c_1[\xi \mapsto c_2] \]

\end{description}

\subsection{Partial Updates}

The presence of a background class permits to have complex values such as sequences of arbitrary length, trees, graphs, etc. An update may only affect a tiny part of a complex value, but the presence of parallelism, though bounded, in sequential ASMs may cause avoidable conflicts on locations bound to such a complex value. For instance, if two updates in parallel affect only separate subtrees, they could be combined into a single tree update. We therefore add the following partial assignment rule to the definition of sequential ASM rules:

\begin{description}

\item[partial assignment.] Whenever $f \in \Sigma$ has arity $n$, $op \in V_K$ is an operator (i.e. a function symbol defined in the background) of arity $m+1$, $t_i$ ($i=1,\dots,n$) and $t_i^\prime$ ($i=1,\dots,m$) are terms over $\Sigma$, then $f(t_1,\dots,t_n) \leftleftarrows^{op} t_1^\prime ,\dots, t_m^\prime$ is a rule.

\end{description}

The effect of a single partial assignment can be captured by a single update
\[ ((f,(\text{val}_S(t_1) ,\dots, \text{val}_S(t_n))), \text{op}(\text{val}_S(f(t_1,\dots,t_n)),\text{val}_S(t_1^\prime) ,\dots, \text{val}_S(t_m^\prime))) \; . \]

However, updates concerning the same location $\ell$ produced by partial assignments are first collected in a multiset $\ddot{\Delta}_\ell$. More precisely, the operators $op$ and their arguments $v_1^\prime ,\dots,v_m^\prime$ (with $v_i^\prime = \text{val}_S(t_i^\prime)$) will be collected in $\ddot{\Delta}_\ell$, i.e. $\ddot{\Delta}_\ell$ takes the form 
\[ \langle (\ell,op_1,(v_1^1 ,\dots,v_{m_1}^1)) , \dots , (\ell,op_k,(v_1^k ,\dots,v_{m_k}^k)) \rangle \; .\]
In \cite{schewe:ejc2011} any member of an update multiset $\ddot{\Delta}_\ell$ of the form $(\ell,op,(v_1 ,\dots,v_{m}))$ is called a {\em shared update}.

If possible, i.e. if the operators and arguments are compatible with each other, this multiset together with $\text{val}_S(f(t_1,\dots,t_n))$ will be collapsed into a single update $(\ell,v_0)$. Conditions for compatibility and the collapse of an update multiset into an update set have been elaborated in detail in \cite{schewe:ejc2011}. 

Furthermore, with bulk values such as trees it is advisable to consider also fragments of values in connection with {\em sublocations}, which leads to dependencies among updates. For a tree $t$ every node $o \in \mathcal{O}$ defines such a sublocation, i.e. a nullary function symbol. In the next section we will make this explicit by providing functions that map values $o \in \mathcal{O}$ to such function symbols and vice-versa. Clearly, the value associated with the root determines the values associated with all sublocations. This is exploited in \cite{schewe:ejc2011} for the analysis of compatibility beyond the permutation of operators, but also on the level of sublocations. For this the notions of \emph{subsumption} and \emph{dependence} between locations are decisive.

\begin{definition}\label{def-subsumption}\rm

A location $\ell_1$ \emph{subsumes} a location $\ell_2$ (notation: $\ell_2 \sqsubseteq \ell_1$) iff for all states $S$ $val_S(\ell_1)$ uniquely determines $val_S(\ell_2)$. A location $\ell_1$ \emph{depends} on a location $\ell_2$ (notation: $\ell_2 \unlhd \ell_1$) iff $val_S(\ell_2) = \bot$ implies $val_S(\ell_1) = \bot$ for all states $S$.

\end{definition}

Clearly, for locations $\ell_1, \ell_2$ with $\ell_2 \sqsubseteq \ell_1$ we also have $\ell_1 \unlhd \ell_2$.

\section{Reflective Sequential Abstract State Machines}\label{sec:rasm}

In this section we define a model of rsASMs extending sequential ASMs. The ground idea is quite simple. It uses a self representation of an ASM, i.e. its signature and rule, as a particular tree value that is assigned to a location \textit{self\/}. For this we exploit the tree algebra from the previous section. Then in every step the update set will be built using the rule in this representation, for which we exploit \textit{raise\/} and \textit{drop\/} as in \cite{stemple:2000}. We conclude the section by giving examples of specifications of reflective sequential algorithms by means of rsASMs.

\subsection{Self Representation Using Trees}\label{ssec:treerep}

For the dedicated location storing the self-representation of a sequential ASM it is sufficient to use a single function symbol \textit{self\/} of arity $0$. Then in every state $S$ the value $\text{val}_S(\textit{self\/})$ is a complex tree comprising two subtrees for the representation of the signature and the rule, respectively. The signature is just a list of function symbols, each having a name and an arity. The rule can be represented by a syntax tree. 

In detail, in the tree structure we have a root node $o$ labelled by \texttt{self} with exactly two successor nodes, say $o_0$ and $o_1$, labelled by \texttt{signature} and \texttt{rule}, respectively. So we have $o \prec_c o_0$, $o_0 \prec_s o_1$ and $o \prec_c o_1$. The subtree rooted at $o_0$ has as many children $o_{00} ,\dots, o_{0k}$ as there are function symbols in the signature, each labelled by \texttt{func}. Each of the subtrees rooted at $o_{oi}$ takes the form $\texttt{func} \langle \texttt{name} \langle f \rangle \; \texttt{arity} \langle n \rangle \rangle$ with a function name $f$ and a natural number $n$. The subtree rooted at $o_1$ represents the rule of a sequential ASM as a tree. Trees representing rules are inductively defined as follows:

\begin{itemize}

\item An assignment rule $f(t_1,\dots,t_n) = t_0$ is represented by a tree of the form 
\[ \textit{label\_hedge}(\texttt{update},\texttt{func} \langle f \rangle \texttt{term} \langle t_1 \dots t_n \rangle \texttt{term} \langle t_0 \rangle) \; . \]

\item A branching rule \texttt{IF} $\varphi$ \texttt{THEN} $r_1$ \texttt{ELSE} $r_2$ \texttt{ENDIF} is represented by a tree of the form 
\[ \textit{label\_hedge}(\texttt{if},\texttt{bool} \langle \varphi \rangle \texttt{rule} \langle t_1 \rangle \texttt{rule} \langle t_2 \rangle) \; , \]
where $t_i$ (for $i=1,2$) is the tree representing the rule $r_i$.

\item A parallel rule \texttt{PAR} $r_1 \dots r_k$ \texttt{ENDPAR} is represented by a tree of the form 
\[ \textit{label\_hedge}(\texttt{par}, \texttt{rule} \langle t_1 \rangle \; \dots \; \texttt{rule} \langle t_k \rangle) \; , \]
where $t_i$ (for $i=1,\dots,k$) is the tree representing the rule $r_i$.

\item A let rule \texttt{LET} $x = t$ \texttt{IN} $r$ is represented by a tree of the form 
\[ \textit{label\_hedge}(\texttt{let}, \texttt{term} \langle x \rangle \texttt{term} \langle t \rangle \texttt{rule} \langle t^\prime \rangle) \; , \]
where $t^\prime$ is the tree representing the rule $r$.

\item A partial assignment rule $f(t_1,\dots,t_n) \leftleftarrows^{op} t_1^\prime, \dots, t_m^\prime$ is represented by a tree of the form
\[ \textit{label\_hedge}(\texttt{partial},\texttt{func} \langle f \rangle \texttt{func} \langle op \rangle \texttt{term} \langle t_1 \dots t_n \rangle \texttt{term} \langle t_1^\prime \dots t_m^\prime \rangle) \; . \]

\end{itemize}

\subsection{The Background of an rsASM}

Let us draw some consequences from this tree representation. As function names in the signature appear in the tree representation, these are values. Furthermore, we may in every step enlarge the signature, so there must be an infinite reserve $\Sigma_{res}$ of such function names. Likewise we require natural numbers in the background for the arity assigned to function symbols, though operations on natural numbers are optional. Terms built over the signature and a base set $B$ must also become values. Concerning the subtree capturing the rule, the keywords for the different rules become labels, so we obtain the set of labels
\begin{align*}
L \qquad = \qquad &\{ \texttt{self},  \texttt{signature}, \texttt{rule}, \texttt{func}, \texttt{name}, \texttt{arity}, \texttt{update}, \texttt{term}, \texttt{if},\\
&\qquad \texttt{bool}, \texttt{par}, \texttt{let}, \texttt{partial} \} .
\end{align*}
Consequently, we must extend a base set $B$ by such terms, i.e. terms will become values.

\begin{definition}\label{def-extended-baseset}\rm

Let $B$ be a base set. An {\em extended base set} is the smallest set $B_{\textit{ext\/}}$ containing $B$ that is closed under adding function symbols in the reserve $\Sigma_{\textit{res\/}}$, natural numbers, the terms $\mathbb{T}$ with respect to $B$ and $\Sigma_{\textit{res\/}}$, and terms of the tree algebra defined over $\Sigma_{\textit{res\/}}$ with labels in $L$ as defined above.

\end{definition}

In an extended base set terms are treated as values that can appear as values of some locations in a state. This implies that terms now have a dual character. When appearing in an ASM rule, e.g. on the right-hand side of an assignment, they are interpreted in the current state to determine an update. However, if they are to be treated as a value, they have to be interpreted by themselves. Therefore, we require a function \textit{drop\/} turning a term into a value, and inversely a function \textit{raise\/} turning a value into a term. On constants in a base set $B$ both functions are just the identity. Note that functions \textit{drop\/} and \textit{raise\/} capture the essence of linguistic reflection\footnote{In a remark in \cite[p.87]{gurevich:tocl2000} Gurevich wrote that for ``algorithms which change their programs during the computation \dots the so-called program is just a part of the data. The real program changes that part of the data, and the real program does not change''. This is only partially correct, as the ``real program'' still has to provide the means to interpret data as executable programs and vice versa, which is exactly what \textit{raise\/} and \textit{drop\/} enable, but these are not covered by sequential algorithms.} \cite{stemple:2000}.

For instance, when evaluating \textit{self\/} in a state $S$ the result should be a tree value, to which we may apply some tree operators to extract a rule associated with some subtree. As this is a value in $B_{ext}$ we may apply \textit{raise\/} to it to obtain the ASM rule, which could be executed to determine an update set and to update the state. Analogously, when assigning a new term (e.g. a Boolean term in a branching rule) to a subtree of \textit{self\/} the value on the right-hand side must be the result of the function \textit{drop\/}, otherwise the term would be evaluated and a Boolean value would be assigned instead.

The functions \textit{drop\/} and \textit{raise\/} can be applied to function names as well, so they can be used as values stored within \textit{self\/} and used as function symbols in rules. In particular, if $\mathcal{O}$ denotes the set of nodes of a tree, then each $o \in \mathcal{O}$ is a value in the base set, but $\textit{raise\/}(o)$ denotes a nullary function symbol that is bound in a state to the subtree $\hat{o}$. However, as it is always clear from the context, when a function name $f$ is used as a value, i.e. as $\textit{drop\/}(f)$, this subtle distinction can be blurred.

Finally, the self-representation as defined above involves several non-logical constants such as the keywords for the rules, i.e. labels in $L$. For theoretical analysis it is important to extract from the representation the decisive terms defined over $\Sigma$ and $B$. That is, we further require an {\em extraction function}\footnote{This extraction function $\beta$ will be used in the formulation of the (reflective) bounded exploration postulate.} $\beta: \mathbb{T}_{ext} \rightarrow \bigcup_{n \in \mathbb{N}} \mathbb{T}^n$, which assigns to each term defined over the extended signature $\Sigma_{ext}$ and the extended base set $B_{ext}$ a tuple of terms in $\mathbb{T}$ defined over $\Sigma$ and $B$. We will see that such an extraction function can actually be derived using the tree algebra. 

The following definition summarises all requirements for the background of an rsASM.

\begin{definition}\label{def-rsasm-background}\rm

The {\em background of an rsASM} is defined by a background class $\mathcal{K}$ over a background signature $V_K$. It must contain an infinite set \textit{reserve\/} of reserve values and an infinite set $\Sigma_{res}$ of reserve function symbols, the equality predicate, the undefinedness value {\em undef}, and a set of labels 
\begin{align*}
L \qquad = \qquad &\{ \texttt{self},  \texttt{signature}, \texttt{rule}, \texttt{func}, \texttt{name}, \texttt{arity}, \texttt{update}, \texttt{term}, \texttt{if},\\
&\qquad \texttt{bool}, \texttt{par}, \texttt{let}, \texttt{partial} \} \; .
\end{align*}
The background class must further define truth values and their connectives, tuples and projection operations on them, natural numbers and operations on them, trees in $T_L$ and tree operations, and the function $\mathbf{I}$, where $\mathbf{I} x . \varphi$ denotes the unique $x$ satisfying condition $\varphi$.

The background must further provide functions: $\textit{drop\/}: \hat{\mathbb{T}}_{ext} \rightarrow B_{ext}$ and $\textit{raise\/}: B_{ext} \rightarrow \hat{\mathbb{T}}_{ext}$ for each base set $B$ and extended base set $B_{ext}$, and a derived {\em extraction function} $\beta: \mathbb{T}_{ext} \rightarrow \bigcup_{n \in \mathbb{N}} \mathbb{T}^n$, which assigns to each term defined over the extended signature $\Sigma_{ext}$ and the extended base set $B_{ext}$ a tuple of terms in $\mathbb{T}$ defined over $\Sigma$ and $B$.

\end{definition}

In the definition we use $\hat{\mathbb{T}}_{ext}$ to denote the union of the set $\mathbb{T}_{ext}$ of terms with $\Sigma_{ext}$ and the set of rules. The extraction function $\beta$ on rule terms is easily defined as follows:
\begin{align*}
& \beta( \textit{label\_hedge}(\texttt{update},\texttt{func} \langle f \rangle \texttt{term} \langle t_1 \dots t_n \rangle \texttt{term} \langle t_0 \rangle) ) = (t_0, t_1, \dots, t_n) \\[1ex]
& \beta( \textit{label\_hedge}(\texttt{if},\texttt{bool} \langle \varphi \rangle \texttt{rule} \langle t_1 \rangle \texttt{rule} \langle t_2 \rangle) ) = (\varphi, t_1^1 ,\dots, t_1^{n_1}, t_2^1 ,\dots, t_2^{n_2}) \\
& \qquad\qquad \text{for} \; \beta(t_1) =(t_1^1 ,\dots, t_1^{n_1}) \;\text{and}\; \beta(t_2) =(t_2^1 ,\dots, t_2^{n_2}) \\[1ex]
& \beta( \textit{label\_hedge}(\texttt{par}, \texttt{rule} \langle t_1 \rangle \;\dots\; \texttt{rule} \langle t_k \rangle) ) = (t_1^1 ,\dots, t_1^{n_1}, \dots, t_k^1 ,\dots, t_k^{n_k}) \\
& \qquad\qquad \text{for}\; \beta(t_i) =(t_i^1 ,\dots, t_i^{n_i}), \; i=1,\dots,k \\[1ex]
& \beta( \textit{label\_hedge}(\texttt{let}, \texttt{term} \langle x \rangle \texttt{term} \langle t \rangle \texttt{rule} \langle t^\prime \rangle) ) = (t, t_1 ,\dots, t_n) \\
& \qquad\qquad \text{for} \; \beta(t^\prime[x \mapsto t]) = (t_1 ,\dots, t_n) \\[1ex]
& \beta( \textit{label\_hedge}(\texttt{partial},\texttt{func} \langle f \rangle \texttt{func} \langle op \rangle \texttt{term} \langle t_1 \dots t_n \rangle \texttt{term} \langle t_1^\prime \dots t_m^\prime \rangle) ) = \\
& \qquad\qquad\qquad\qquad (t_1 ,\dots, t_n, op(f(t_1 ,\dots, t_n), t_1^\prime, \dots, t_m^\prime))
\end{align*}
Definition \ref{def-rsasm-background} allows us to define rsASMs realising the ideas for reflective computing.

\begin{definition}\label{def-rsasm}\rm

A {\em reflective sequential ASM} (rsASM) $\mathcal{M}$ comprises an (initial) signature $\Sigma$ containing a $0$-ary function symbol \textit{self\/}, a background as defined in Definition \ref{def-rsasm-background}, and a set $\mathcal{I}$ of initial states over $\Sigma$ closed under isomorphisms such that any two states $I_1, I_2 \in \mathcal{I}$ coincide on \textit{self\/}. Furthermore, $\mathcal{M}$ comprises a state transition function $\tau$ on states over extended signature $\Sigma_S$ with $\tau(S) = S + \Delta_{r_S}(S)$, where the rule $r_S$ is defined as $\textit{raise\/}(\textit{rule\/}(\text{val}_S(\textit{self\/})))$ over the signature $\Sigma_S = \textit{raise\/}(\textit{signature\/}(\text{val}_S(\textit{self\/})))$ .

\end{definition}

In this definition we use extraction functions \textit{rule\/} and \textit{signature\/} defined on the tree representation of a sequential ASM in \textit{self\/}. These are simply defined as
\begin{gather*}
\textit{signature\/}(t) = \textit{subtree\/}(\textbf{I} o . \textit{root}(t) \prec_c o \wedge \textit{label\/}(o) = \texttt{signature} \qquad \text{and} \\
\textit{rule\/}(t) = \textit{subtree\/}(\textbf{I} o . \textit{root}(t) \prec_c o \wedge \textit{label\/}(o) = \texttt{rule}
\end{gather*}

\subsection{Difference of Trees}

We know that for every state $S$ there is a well-defined, consistent update set $\Delta(S)$ such that $S^\prime = S+\Delta(S)$ is the successor state of $S$. Actually, $\Delta(S)$ arises as the difference between $S$ and $S^\prime$. If the states contain locations with bulk values assigned to them, then it becomes also important to provide means for the expression of the difference of such values. 

Here we concentrate only on the tree values assigned to \textit{self\/}. Let $T_L$ be the set of trees with labels in $L$ and values in a universe $\mathcal{U}$ as used in Definition \ref{def-rsasm-background}.

\begin{proposition}\label{prop-tree-difference}

For trees $t, t^\prime \in T_L$ we can write $t^\prime = \theta(t)$ with a tree algebra term $\theta$.

\end{proposition}

\begin{proof}
First consider the subtrees $t_{sig}$ and $t_{sig}^\prime$ representing the signatures in $S$ and $\tau(S)$, respectively. As we assume that only new function symbols are added, we obtain immediately
\[ t_{sig}^\prime = \textit{right\_extend\/}(t_{sig}, \textit{label\_hedge\/}(\texttt{func}, \langle f_1 \rangle \, \langle a_1 \rangle) \dots \textit{label\_hedge\/}(\texttt{func}, \langle f_k \rangle \, \langle a_k \rangle)) . \]

Concerning the subtrees $t_{rule}$ and $t_{rule}^\prime$ representing the rules in $\tau(S)$ and $\tau(S)$, respectively, we proceed by structural induction from the leaves to the root, for which it is sufficient to consider the following four cases:

\begin{enumerate}

\item If $o$ is a node in $t_{rule}^\prime$ representing an assignment rule that does not appear in $t$, then $\textit{subtree\/}(o)$ takes the form
\[ \textit{subtree\/}(o) = \textit{label\_hedge}(\texttt{update},\texttt{func} \langle f \rangle \texttt{term} \langle t_1 \dots t_n \rangle \texttt{term} \langle t_0 \rangle) . \]

\item If $o$ is a node in $t_{rule}^\prime$ representing a partial assignment rule that does not appear in $t$, then $\textit{subtree\/}(o)$ takes the form
\[ \textit{subtree\/}(o) = \textit{label\_hedge}(\texttt{partial},\texttt{func} \langle f \rangle \texttt{func} \langle op \rangle \texttt{term} \langle t_1 \dots t_n \rangle \texttt{term} \langle t_1^\prime \dots t_m^\prime \rangle) . \]

\item Let $o_{kl}$ denote the node in $t$ with $\textit{root\/}(t) \prec^k o_{kl}$ such that there are nodes $o_1 ,\dots, o_l$ with $o_1 \prec_s o_2 \prec_s \dots \prec_s o_l \prec_s o_{kl}$, but no node $o_0$ with $o_0 \prec_s o_1$. If $\textit{subtree\/}(o_{kl}$ appears as a subtree in $t^\prime$, say at node $o^\prime$, we simple have $\textit{subtree\/}(o^\prime) = \textit{subtree\/}(o_{kl})$.

\item If $o$ is a node in $t_{rule}^\prime$ with children $o_1 ,\dots, o_l$ and we can write $\textit{subtree\/}(o_i) = \theta_i(t)$, then we obtain 
\[ \textit{subtree\/}(o) = \textit{label\_hedge}(a, \theta_1(t) \dots \theta_l(t)) \]

with a label $a \in \{ \texttt{if}, \texttt{par}, \texttt{let} \}$.

\end{enumerate}

Using cases (1), (2) and (3) with a maximum tree associated with $o^\prime$ that already appears as a subtree of $t$ gives our induction base. Using case (4) defines the induction step.
\end{proof}

\begin{corollary}\label{prop-tree-update}

For trees $t, t^\prime \in T_L$ there exists a compatible update multiset $\ddot{\Delta}$ defined by updates and shared updates on the nodes of $t^\prime$ such that $\ddot{\Delta}$ collapses to $\Delta = \{ (\textit{root\/}(t^\prime), t^\prime) \}$.

\end{corollary}

\begin{proof}
For the subtree representing the signature we obtain a rule of the form let $o^\prime$ be the unique successor node of the root of $t^\prime$ labelled by \texttt{signature}. This gives rise to the rule
\[ o^\prime := \textit{right\_extend\/}(t_{sig}, \textit{label\_hedge\/}(\texttt{func}, \langle f_1 \rangle \, \langle a_1 \rangle) \dots \textit{label\_hedge\/}(\texttt{func}, \langle f_k \rangle \, \langle a_k \rangle)) , \]
where the subtree $t_{sig}$ represents the signatures in $S$.

For the subtree representing the rule we simply define rules for the four cases used in the proof of Proposition \ref{prop-tree-difference}:

\begin{enumerate}

\item If $o$ is a node in $t_{rule}^\prime$ representing an assignment rule that does not appear in $t$, the rule takes the form
\[ o := \textit{label\_hedge}(\texttt{update},\texttt{func} \langle f \rangle \texttt{term} \langle t_1 \dots t_n \rangle \texttt{term} \langle t_0 \rangle) . \]

\item If $o$ is a node in $t_{rule}^\prime$ representing a partial assignment rule that does not appear in $t$, then the rule takes the form
\[ o := \textit{label\_hedge}(\texttt{partial},\texttt{func} \langle f \rangle \texttt{func} \langle op \rangle \texttt{term} \langle t_1 \dots t_n \rangle \texttt{term} \langle t_1^\prime \dots t_m^\prime \rangle) . \]

\item In case $o$ is a node in $t_{rule}^\prime$ such that $\textit{subtree\/}(o)$ occurs as a subtree of $t$, the rule takes the form 
\begin{gather*}
\texttt{LET}\; o_{kl} = \textbf{I} o^\prime . \textit{root\/}(t) \prec_c^k o^\prime \wedge \exists o_1,\dots,o_l . o_1 \prec_s o_2 \prec_s \dots \prec_s o_l \prec_s o^\prime \wedge \\
\neg \exists o_0 . o_0 \prec_s o_1 \;\texttt{IN}\; o := \textit{subtree\/}(o_{kl})
\end{gather*}

\item For all other nodes $o$ the rule takes the form $o := \textit{label\_hedge}(a, t_1 \dots t_l)$ with a label $a \in \{ \texttt{if}, \texttt{par}, \texttt{let} \}$ and terms $t_1 ,\dots, t_l$ that are used in the assignment rules for the children of $o$.

\end{enumerate}

The parallel composition of these rules for all nodes of $t^\prime$ defines a rule $r$ that yields the required update multiset and update set.
\end{proof}

\subsection{Examples}

We first show an example for the use of the tree algebra, then several examples of reflective sequential algorithms, specified by rsASMs. Example \ref{bsp-join} addresses the well known example of a natural join, where reflection is required to compute the type of the result, and Example \ref{bsp-src} shows how parity of sets can be addressed by means of reflection\footnote{As the algorithms investigated in this paper are sequential, we only deal with the parity of subsets of a fixed finite set.}.

\begin{example}\label{bsp-tree}

To illustrate the tree algebra operations consider a tree $t$ with a root labelled by \texttt{self}. Let there be two direct children of the root labelled with \texttt{signature} and \texttt{rule}, respectively, and let the subtree labelled by signature have the form $\textit{label\_hedge\/}(\texttt{signature},h)$, where the hedge $h$ is a sequence of trees of the form $\textit{label\_hedge\/}(\texttt{func},\langle f \rangle \; \langle a \rangle)$, where $f$ is the name of a function symbol in $\Sigma$ and $a$ denotes the arity of $f$. Then the ASM rule
\begin{align*}
\texttt{LET}\; &\text{sign} = \mathbf{I} o . \textit{root\/}(t) \prec_c o \wedge \textit{label\/}(o) = \texttt{signature} \;\texttt{IN} \\
\texttt{LET}\; &h = \mathbf{I} h_s . \textit{subtree\/}(o) = \textit{label\_hedge\/}(\texttt{signature}, h_s) \;\texttt{IN} \\
& t := \textit{inject\_hedge\/} ( \textit{context\/}(\textit{root\/}(t), \text{sign}), \\
& \hspace*{1.5cm} \textit{label\_hedge\/}(\texttt{signature}, \textit{concat\/}(h,\textit{label\_hedge\/}(\texttt{func},\langle f \rangle \langle a \rangle))))
\end{align*}

inserts a new function symbol $f$ with arity $a$ into the tree representation. We used the notation $\mathbf{I} x. \varphi$ to denote the unique $x$ satisfying $\varphi$. We might use this rule as a definition for \texttt{NewFunction}($f$,$a$).

\end{example}

\begin{example}\label{bsp-join}
For the background let us assume the presence of a finite set $\mathcal{D}$ of domain values and a finite set $\mathcal{A}$ of attributes. In an initial state $S_0$ we assume $k$ relational function symbols $R_1 ,\dots, R_k$ with arities $n_1 ,\dots, n_k$ in the signature, i.e. we have $\text{val}_{S_0}(R_i,(v_1 ,\dots, n_i)) \in \{ \textbf{true},\; \textbf{false} \}$, if $v_i \in \mathcal{D}$ holds for all $i=1,\dots, n_i$, and $\text{val}_{S_0}(R_i,(v_1 ,\dots, n_i)) = \textit{undef\/}$ otherwise. 

The attributes $A \in \mathcal{A}$ are linked to the function symbols $R_i$ using a binary function symbol \textit{index\/} that is initially defined on $\{ \textit{drop\/}(R_i) \mid 1 \le i \le k \} \times \mathcal{A}$, i.e. on pairs comprising the name of a function symbol and attribute, such that $\textit{index\/}(\textit{drop\/}(R_i), A) \in \{ 1 ,\dots, n_i \}$. Further assume that for all $i$ \textit{index\/} is injective on $\{ \textit{drop\/}(R_i) \} \times \mathcal{A}$, and that for each $j \in \{ 1 ,\dots, n_i \}$ there exists some $A \in \mathcal{A}$ with $\textit{index\/}(\textit{drop\/}(R_i), A) = j$.

With this we also obtain (derived) ``projection'' functions $\hat{R}_i$ of arity $n_i + 1$ (for $1 \le i \le k$). We have $\text{val}_{S_0}(\hat{R}_i,(A,v_1 ,\dots, v_{n_i})) = v_j$, if $\textit{index\/}(\textit{drop\/}(R_i), A) = j$ and $\text{val}_{S_0}(R_i,(v_1 ,\dots, v_{n_i})) = \textbf{true}$ hold, otherwise we obtain $\textit{undef\/}$.

This gives rise to $\text{val}_{S_0}(\textit{self\/}) = t_0$ with a tree value $t_0 = \textit{label\_hedge\/}(\texttt{self}, t_0^{sig} \; t_0^{rule})$ with $t_0^{sig} = \textit{label\_hedge\/}(\texttt{signature}, t_0^0 \, t_0^1 \, \hat{t}_0^1 ,\dots, t_0^k \, \hat{t}_0^k)$ with
\begin{gather*}
t_0^0 = \textit{label\_hedge\/}(\texttt{func}, \texttt{name}\langle \textit{index\/} \rangle \; \texttt{arity}\langle 2 \rangle) , \\
t_0^i = \textit{label\_hedge\/}(\texttt{func}, \texttt{name}\langle R_i \rangle \; \texttt{arity}\langle n_i \rangle) , \;\text{and}\\
\hat{t}_0^i = \textit{label\_hedge\/}(\texttt{func}, \texttt{name}\langle \hat{R}_i \rangle \; \texttt{arity}\langle n_i + 1 \rangle).
\end{gather*}

We further have $t_0^{rule} = \textit{label\_hedge\/}(\texttt{rule}, t_r)$ with a tree $t_r$ representing a rule $r$ that we define as \texttt{PAR} $r_{init}$ $r_{join}$ \texttt{ENDPAR}. For the definition of the rules $r_{init}$ and $r_{join}$ we use an additional control-state variable \textit{mode\/}, i.e. a function symbol of arity $0$  in the signature. The rule $r_{init}$ can be defined as follows\footnote{Note that the set comprehension assigned to $t_i$ and $t_j$ can be avoided using a \texttt{PAR}-rule combining rules that parameterised by the fixed set $\mathcal{A}$.}:

\begin{tabbing}
xxx\=xxxxx\=xxxxx\=xxxxx\=xxxxx\=xxxxx\=xxxxx\= \kill
\> \texttt{IF} \> \textit{mode\/} = init \\
\> \texttt{THEN} \> \texttt{LET} \> $J_{ij}$ = NewFunc, \\
\>\>\> $\hat{J}_{ij}$ = NewFunc \texttt{IN} \\
\>\> \texttt{LET} \> $t_i = \{ A \in \mathcal{A} \mid \textit{index\/}(\textit{drop\/}(R_i), A) \neq \textit{undef\/} \}$ , \\
\>\>\> $t_j = \{ A \in \mathcal{A} \mid \textit{index\/}(\textit{drop\/}(R_j), A) \neq \textit{undef\/} \}$ \texttt{IN} \\
\>\> \texttt{PAR} \> \texttt{LET} $n = \textit{card\/}(t_i \cup t_j)$ \texttt{IN} \\
\>\>\> $\textbf{I} o . (\textit{root\/}(\textit{self\/}) \prec_c o \wedge \textit{label\/}(o) = \texttt{signature}) \leftleftarrows^{\textit{right\_extend\/}}$ \\
\>\>\>\>\> $\textit{label\_hedge\/}(\texttt{func}, \texttt{name} \langle J_{ij} \rangle \; \texttt{arity} \langle n \rangle)$ ,\\
\>\>\>\>\> $\textit{label\_hedge\/}(\texttt{func}, \texttt{name} \langle \hat{J}_{ij} \rangle \; \texttt{arity} \langle n+1 \rangle)$ \\
\>\>\> \dots $\textit{new\_index\/}(J_{ij}, A)$ \dots \% for all $A \in \mathcal{A}$ \\
\>\>\> \textit{mode\/} := join \\
\>\> \texttt{ENDPAR} \\
\> \texttt{ENDIF}
\end{tabbing}

where the rule $\textit{new\_index\/}(J_{ij}, A)$ for $A \in \mathcal{A}$ is defined as\footnote{Note that we use here $n_i$, i.e. the arity of the function symbol $R_i$. This can be defined using a \texttt{LET}-construct and $\mathbf{I}$, by means of which the value is extracted from the subtree representing the signature. We omit the details how to write this term.}

\begin{tabbing}
xxx\=xxxxx\=xxxxx\=xxxxx\=xxxxx\=xxxxx\=xxxxx\= \kill
\> \texttt{IF} \> $A \in t_i$ \\
\> \texttt{THEN} \> $\textit{index\/}(J_{ij}, A) := \textit{index\/}(\textit{drop\/}(R_i), A)$ \\
\> \texttt{ELSE} \> \texttt{IF} \> $A \in t_j - t_i$ \\
\>\> \texttt{THEN} \> $\textit{index\/}(J_{ij}, A) := n_i + \textit{index\/}(\textit{drop\/}(R_j), A) -$\\
\>\>\>\quad $\textit{card\/}(\{ B \in t_i \cap t_j \mid \textit{index\/}(\textit{drop\/}(R_j), B) < \textit{index\/}(\textit{drop\/}(R_j), A) \})$ \\
\>\> \texttt{ENDIF} \\
\> \texttt{ENDIF}
\end{tabbing}

The rule $r_{join}$ can be defined as \texttt{PAR} \dots $\textit{join\/}(x_1,\dots,x_n)$ \dots (for all $(x_1,\dots,x_n) \in \mathcal{D}^n$) \textit{mode\/} := halt \texttt{ENDPAR}, where the rules $\textit{join\/}(x_1,\dots,x_n)$ are defined as follows:

\begin{tabbing}
xxx\=xxxxx\=xxxxx\=xxxxx\=xxxxx\=xxxxx\=xxxxx\= \kill
\> \texttt{IF} \> \textit{mode\/} = join \\
\> \texttt{THEN} \> \texttt{LET} \> $t_i = \{ A \in \mathcal{A} \mid \textit{index\/}(\textit{drop\/}(R_i), A) \neq \textit{undef\/} \}$ , \\
\>\>\> $t_j = \{ A \in \mathcal{A} \mid \textit{index\/}(\textit{drop\/}(R_j), A) \neq \textit{undef\/} \}$ \texttt{IN} \\
\>\> \texttt{LET} \> $n = \textit{card\/}(t_i \cup t_j)$ \texttt{IN} \\
\>\> \texttt{LET} \> \dots \% for all $k$ with $1 \le k \le n_i + n_j -n$ \% \\
\>\>\>\> $\hat{x}_k = \textbf{I} x . \bigwedge_{1 \le l \le n_i} \exists A. \textit{index\/}(\textit{drop\/}(R_j), A) = k \wedge$ \\
\>\>\>\>\> $\textit{index\/}(\textit{drop\/}(R_i), A) = l \Rightarrow x = x_l$ \dots \texttt{IN} \\
\>\> \texttt{IF} \> $R_i(x_1,\dots,x_{n_i}) = \textbf{true} \wedge$ \\
\>\>\>\>\> $R_j(\hat{x}_1 ,\dots, \hat{x}_{n_i+n_j-n}, x_{n_i+1}, \dots, x_n) = \textbf{true}$ \\
\>\> \texttt{THEN} \> \texttt{PAR} \> $J_{ij}(x_1,\dots,x_n) := \textbf{true}$ \\
\>\>\>\> \dots $\textit{join\/}^\prime(A,x_1,\dots,x_n)$ \dots \% for all $A \in \mathcal{A}$ \% \\
\>\>\> \texttt{ENDPAR} \\
\>\> \texttt{ENDIF} \\
\> \texttt{ENDIF}
\end{tabbing}

where the rule $\textit{join\/}^\prime(A,x_1,\dots,x_n)$ (for $A \in \mathcal{A}$ and $x_1,\dots,x_n \in \mathcal{D}$) is defined as \texttt{IF} $\textit{index\/}(J_{ij}, A) = j$ \texttt{THEN} $\hat{J}_{ij}(A,x_1,\dots,x_n) := x_j$ \texttt{ENDIF}.

\end{example}

Note that in this example we merely used reflection to introduce new function symbols. Reflection is required to determine the ``type'', i.e. the set of attributes, of the join relation. Of course, the example makes more sense, if it is embedded into a full specification of ``type-safe'' relational algebra, where it is known that only the natural join cannot be expressed by parametric polymorphism.

\begin{example}\label{bsp-src}
Let us assume a fixed finite set $\mathcal{D}$, defined in the background. Besides $\textit{self\/}$ the signature comprises 0-ary functions symbols \textit{card\/}, \textit{parity\/} and \textit{mode\/}---the latter one set to init in an initial state---and a unary function symbol \textit{set\/}. In an initial state we require $\textit{set\/}(x) \in \{ \textbf{true}, \textbf{false} \}$ for all $x \in \mathcal{D}$ and $\textit{set\/}(x) = \textit{undef\/}$ for all $x \notin \mathcal{D}$. Locations with function symbol \textit{set\/} will never be altered. The rule represented in \textit{self\/} in an initial state takes the following form:

\begin{tabbing}
xxx\=xxxxx\=xxxxx\=xxxxx\=xxxxx\=xxxxx\=xxxxx\= \kill
\> \texttt{IF} \> \textit{mode\/} = init \\
\> \texttt{THEN} \> \texttt{PAR} \> \textit{card\/} $:= 0$ \\
\>\>\> $\dots \textit{count\/}(x) \dots$ \% for all $x \in \mathcal{D}$ \\
\>\>\> \textit{mode\/} $:=$ count \\
\>\> \texttt{ENDPAR} \\
\> \texttt{ELSE} \> \texttt{IF} \> \textit{mode\/} = count \\
\>\> \texttt{THEN} \> \texttt{PAR} \> \textit{mode\/} $:=$ eval \texttt{ENDPAR} \\
\>\> \texttt{ELSE} \> \texttt{IF} \> \textit{mode\/} = eval \\
\>\>\> \texttt{THEN} \> \texttt{PAR} \> \textit{parity\/} $:=$ \textit{card\/} mod 2 \\
\>\>\>\>\> \textit{mode\/} $:=$ halt \\
\>\>\>\> \texttt{ENDPAR} \\
\>\>\> \texttt{ENDIF} \\
\>\> \texttt{ENDIF} \\
\> \texttt{ENDIF}
\end{tabbing}

Here the rules $\textit{count\/}(x)$ for $x \in \mathcal{D}$ are defined as follows\footnote{Note that in this rule we use a partial assignment to $o$, though $o$ is only a logical variable defined by the \texttt{LET}-construct. So, stricly speaking this is not correct. However, in the previous section we remarked that the nodes of a tree are values that define nullary function symbols and consequently sublocations---i.e. $o$ is actually $\textit{raise\/}(o)$---but we wanted to drop this subtle distinction.}

\begin{tabbing}
xxx\=xxxxx\=xxxxx\=xxxxx\=xxxxx\=xxxxx\=xxxxx\= \kill
\> \texttt{IF} \> $\textit{set\/}(x) = \textbf{true}$ \\
\> \texttt{THEN} \> \texttt{LET} \> $o = \mathbf{I} o_{32} . \exists o_1, o_2, o_{31} . \textit{root\/}(\textit{self\/}) \prec_c o_1 \prec_c o_2 \prec_c o_{31} \prec_s o_{32} \wedge$ \\
\>\>\>\>\> $\textit{label\/}(o_1) = \texttt{rule} \wedge \textit{label\/}(o_{31}) = \texttt{bool}$ \\
\>\> \texttt{IN} \> $o \leftleftarrows^{\textit{right\_extend\/}} \textit{label\_hedge\/}(\texttt{partial},$ \\ 
\>\>\>\>\> $\texttt{func}\langle \textit{card\/} \rangle, \texttt{func}\langle + \rangle, \texttt{term}\langle \epsilon \rangle, \texttt{term}\langle 1 \rangle)$ \\
\> \texttt{ENDIF}
\end{tabbing}

Thus, in an initial state the reflective algorithm initialises \textit{card\/} and creates the rule---addition of $1$ to \textit{card\/} for every element in the input set---that will be applied next, when the control-state variable \textit{mode\/} takes the value count. Hence, the cardinality of the input set is computed, and in the final step the parity is derived from it.
\end{example}

Note that in this example the same result could also be achieved by using directly partial updates on \textit{card\/}. In general, however, creating rules for all elements of a set may involve more complex operations. We merely used Example \ref{bsp-src} to illustrate reflection without claiming that the problem at hand can only be solved in this way.

Note also that Example \ref{bsp-src} could be easily extended to compute the parity of an arbitrary finite set with elements in a non-necesaarily finite domain $\mathcal{D}$, if we drop the restriction to sequential algorithms\footnote{This would require to use unbounded parallelism, which in ASMs is supported through \texttt{FORALL}-rules \cite{boerger:2003}.}. Alternatively, the problem could also have been approched using non-deterministic ASMs with \texttt{CHOOSE}-rules.

\section{Axiomatisation of Reflective Sequential Algorithms}\label{sec:rsa}

The celebrated sequential ASM thesis needs only three simple, intuitive postulates to define sequential algorithms (for details see the deep discussion in \cite{gurevich:tocl2000}):

\begin{description}

\item[Sequential time.]  Each sequential algorithm proceeds in sequential time using states, initial states and transitions from states to successor states, i.e. there is a set $\mathcal{S}$ of states, a subset $\mathcal{I} \subseteq \mathcal{S}$ of initial states, and a {\em transition function} $\tau :\mathcal{S} \rightarrow \mathcal{S}$, which maps a state $S \in \mathcal{S}$ to its successor state $\tau(S)$.

\item[Abstract state.] States $S \in \mathcal{S}$ are universal algebras (aka Tarski structures), i.e. functions resulting from the interpretation of a signature $\Sigma$, i.e. a set of function symbols, over a base set. The sets $\mathcal{S}$ and $\mathcal{I}$ of states and initial states, respectively, are closed under isomorphisms. States $S$ and successor states $\tau(S)$ have the same base set, and if $\sigma$ is an isomorphism defined on $S$, then also $\tau(\sigma(S)) = \sigma(\tau(S))$ holds.

\item[Bounded exploration.] There exists a finite set $W$ of ground terms (called {\em bounded exploration witness}) such that the difference between a state and its successor state (called {\em update set}) is uniquely determined by the values of these terms in the state.

\end{description}

In this section we discuss how to modify these postulates in order to define reflective sequentials algorithms (RSAs).

\subsection{Sequential Time}

In principle, also reflective sequential algorithms proceed in sequential time. However, the crucial feature of reflection is that in every step the algorithm may change. So the question is, whether this impacts on the postulate. We argue that it is always possible to have a finite representation of a sequential algorithm, which can therefore be subsumed in the notion of state, while the sequential time postulate can remain unchanged.

Our argument is grounded in the axiomatisation of sequential algorithms and the proof of the sequential ASM thesis in \cite{gurevich:tocl2000}. According to the postulates it suffices to represent a sequential algorithm $P$ by a set of pairs $(S,\Delta_P(S))$ comprising states $S$ and the update set of $P$ in that state. A consequence of the proof of the sequential ASM thesis is that update sets $\Delta_P(S_i)$ ($i=1,2$) are equal, if the states $S_1$ and $S_2$ are $W$-equivalent for a fixed bounded exploration witness $W$. We have $S_1 \sim_W S_2$ iff $E_{S_1} = E_{S_2}$, where $E_S$ is the equivalence relation on $W$ defined by $E_S(t_1,t_2) \equiv val_S(t_1) = val_S(t_2)$, where $val_S(t)$ denotes the interpretation of a ground term $t$ as a value in the base set of a state $S$. It is therefore sufficient to replace the state $S$ by a condition $\varphi_{S}$, which evaluates to true on states that are $W$-equivalent to $S$. As there can only be finitely many $W$-equivalence classes, we obtain an abstract finite representation by a finite set of pairs $(\varphi_i,\Delta_i)$ ($i=1,\dots,k$).

\begin{postulate}[Sequential Time Postulate]\label{p-time}
A {\em reflective sequential algorithm} $\mathcal{A}$ comprises a set $\mathcal{S}$ of {\em states}, a subset $\mathcal{I} \subseteq \mathcal{S}$ of {\em initial states}, and a {\em one-step transition relation} $\tau \subseteq \mathcal{S} \times \mathcal{S}$. Whenever $\tau(S) = S^\prime$ holds, the state $S^\prime$ is called the {\em successor state} of the state $S$.

\end{postulate}

A {\em run} of an RSA $\mathcal{A}$ is then given by a sequence $S_0, S_1, \dots$ of states $S_i \in \mathcal{S}$ with an initial state $S_0 \in \mathcal{I}$ and $S_{i+1} = \tau(S_i)$.

\subsection{Abstract States}

Our argumentation above further justifies to postulate states $S$ of RSAs to be defined by structures over a signature $\Sigma_S$. First recall some basic definitions. 

A {\em signature} $\Sigma$ is a finite set of function symbols, and each $f \in \Sigma$ is associated with an {\em arity} $\text{ar}(f) \in \mathbb{N}$. A {\em structure} over $\Sigma$ comprises a {\em base set} $B$ and an {\em interpretation} of the function symbols $f \in \Sigma$ by functions $f_S : B^{\text{ar}(f)} \rightarrow B$. An {\em isomorphism} $\sigma$ between two structures is given by a bijective mapping $\sigma : B \rightarrow B^\prime$ between the base sets that is extended to the functions by $\sigma(f_B)(\sigma(a_1),\dots,\sigma(a_n)) = \sigma(f_B(a_1,\dots,a_n))$ for all $a_i \in B$ and $n = \text{ar}(f)$. For convenience to capture partial functions we may assume that base sets contain a constant \textit{undef\/} and that each isomorphism $\sigma$ maps \textit{undef\/} to itself.

In order to capture reflection the following changes to the abstract state postulate for sequential algorithms must be taken into account:

\begin{enumerate}

\item As the algorithm that is to be applied in state $S$ is represented in $S$, there must exist a subsignature $\Sigma_{alg} \subseteq \Sigma_S$ for this, and we can assume a function that maps the restriction of $S$ to $\Sigma_{alg}$ to a sequential algorithm $\mathcal{A}(S)$ that operates on states defined over $\Sigma_S$.

\item As only a single RSA is to be considered, it does not make sense to permit arbitrary sequential algorithms $\mathcal{A}(S_0)$ in initial state $S_0 \in \mathcal{I}$. We should therefore require $\mathcal{A}(S_0) = \mathcal{A}(S_0^\prime)$ for all initial states $S_0, S_0^\prime \in \mathcal{I}$. This unique $\mathcal{A}(S_0)$ is then the {\em initial sequential algorithm}.

\item As $\mathcal{A}(S)$ manipulates also locations over $\Sigma_{alg}$ we have in general $\mathcal{A}(\tau(S)) \neq \mathcal{A}(S)$. In particular, the signature may have changed, i.e. $\Sigma_{\tau(S)} \neq \Sigma_S$. Without loss of generality we can assume that only new function symbols will be added, i.e. $\Sigma_S \subseteq \Sigma_{\tau(S)}$.

\item As states comprise representations of sequential algorithms, we cannot simply assume arbitrary base sets $B$, though the restriction of $S$ to $\Sigma_S - \Sigma_{alg}$ should still be a structure with a base set $B$, while for $S$ as a whole we need an extension $B_{ext}$. This extension must at least contain terms over $\Sigma_S$ and $B$. Let us reserve the notation {\em (standard) base set} for the arbitrary $B$, and call $B_{ext}$ an {\em extended base set}.

\item As for the sequential ASM thesis it follows that there is a unique minimal consistent update set $\Delta(S)$ capturing the difference between the state $S$ and its successor $\tau(S)$, and $\tau(S)$ results from applying $\Delta(S)$ to $S$ (denoted as $\tau(S) = S + \Delta(S)$). However, the algorithm yielding $\Delta(S)$ is $\mathcal{A}(S)$, i.e. we have $\tau(S) = S + \Delta_{\mathcal{A}(S)}(S)$.

\end{enumerate}

In the last condition we make use of the unique consistent update set $\Delta(S)$ defined by the algorithm in state $S$. This requires some explanation. As for sequential ASMs a {\em location} $\ell$ in state $S$ is a pair $(f,(v_1,\dots,v_n))$ with a function symbol $f \in \Sigma_S$ of arity $n$ and values $v_1 ,\dots,v_n$ in the (extended) base set. If the interpretation defines $f_S(v_1,\dots,v_n) = v_0$, then the value $v_0 \in B_{ext}$ is called the {\em value} of location $\ell$ in $S$, which we denote as $\text{val}_S(\ell)$. An {\em update} is a pair $(\ell,v)$ comprising a location and a value. An {\em update set} is a set $\Delta$ of updates. An update set is {\em consistent} iff $(\ell,v_1) \in \Delta$ and $(\ell,v_2) \in \Delta$ imply $v_1 = v_2$. As for sequential ASMs we define the state $S^\prime = S + \Delta$ resulting from the {\em application of $\Delta$ to $S$} by
\[ val_{S^\prime}(\ell) = \begin{cases} v &\text{if}\; (\ell,v) \in \Delta \\
\text{val}_S(\ell) &\text{else} \end{cases} \]

provided $\Delta$ is consistent, otherwise we set $S + \Delta = S$. Concerning a state $S$ and its successor $\tau(S)$ we obtain a set $\textit{Diff\/}\, = \{ \ell \mid \text{val}_{\tau(S)}(\ell) \neq \text{val}_S(\ell) \}$ of those locations, where the states differ\footnote{Note that \textit{Diff\/} contains also all locations with new function symbols taken from the reserve.}. Then $\Delta(S) = \{ (\ell,v) \mid \ell \in \textit{Diff\/} \wedge v = \text{val}_{\tau(S)}(\ell) \}$ is a consistent update set with $S + \Delta(S) = \tau(S)$ and furthermore, $\Delta(S)$ is minimal with this property.

Our considerations above lead us to the following modified abstract state postulate for reflective sequential algorithms.

\begin{postulate}[Abstract State Postulate]\label{p-state}
States of a reflective sequential algorithm must satisfy the following conditions:

\begin{enumerate}

\item Each state $S \in \mathcal{S}$ of an RSA $\mathcal{A}$ is a structure over some signature $\Sigma_S$, and an extended base set $B_{ext}$. The extended base set $B_{ext}$ contains at least a {\em standard base set} $B$ and all terms defined over $\Sigma_S$ and $B$.

\item The sets $\mathcal{S}$ and $\mathcal{I}$ of states and initial states, respectively, are closed under isomorphisms. 

\item Whenever $\tau(S) = S^\prime$ holds, then $\Sigma_S \subseteq \Sigma_{\tau(S)}$, the states $S$ and $S^\prime$ have the same standard base set, and if $\sigma$ is an isomorphism defined on $S$, then also $\tau(\sigma(S)) =\sigma(\tau(S))$ holds.

\item There exists a fixed subsignature $\Sigma_{alg} \subseteq \Sigma_S$ for all $S$ and a function that maps the restriction of $S$ to $\Sigma_{alg}$ to a sequential algorithm $\mathcal{A}(S)$ with signature $\Sigma_S$, such that $\tau(S) = S + \Delta_{\mathcal{A}(S)}(S)$ holds for the successor state $\tau(S)$.

\item For all initial states $S_0, S_0^\prime \in \mathcal{I}$ we have $\mathcal{A}(S_0) = \mathcal{A}(S_0^\prime)$.

\end{enumerate}

\end{postulate}

\subsection{Background}

Same as in for sequential algorithms we need some to formulate minimum requirements for the {\em background}, but this time the requirements are too elaborate to leave them implicit. Similar to Definition \ref{def-rsasm-background} these requirements concern the reserve, truth values, tuples, \textit{raise\/} and \textit{drop\/}, but it leaves open how sequential algorithms are represented by structures over $\Sigma_{alg}$. In view of our discussion in the previous subsection the following postulate is thus straightforward.

\begin{postulate}[Background Postulate]\label{p-background}
The {\em background of an RSA} is defined by a background class $\mathcal{K}$ over a background signature $V_K$. It must contain an infinite set \textit{reserve\/} of reserve values and an infinite set $\Sigma_{res}$ of reserve function symbols, the equality predicate, the undefinedness value {\em undef}, truth values and their connectives, tuples and projection operations on them, natural numbers and operations on them, and constructors and operators that permit the representation and update of sequential algorithms. 

The background must further provide functions: $\textit{drop\/}: \mathbb{T}_S \rightarrow B_{ext}$ and $\textit{raise\/}: B_{ext} \rightarrow \mathbb{T}_S$ for each base set $B$ and extended base set $B_{ext}$, and an {\em extraction function} $\beta: \mathbb{T}_S \rightarrow \bigcup_{n \in \mathbb{N}} \mathbb{T}^n$, which assigns to each term defined over a signature $\Sigma_S$ and the extended base set $B_{ext}$ a tuple of terms in $\mathbb{T}$ defined over $\Sigma_S - \Sigma_{alg}$ and $B$.

\end{postulate}

\subsection{Bounded Exploration}

Let us finally discuss a generalisation of the bounded exploration postulate. In principle, an RSA may increase its signature in every step, so a priori it is impossible to find a fixed finite bounded exploration witness that determines update sets in every state. However, in every state $S$ we have a representation of the actual sequential algorithm $\mathcal{A}(S)$ as requested by the abstract state postulate \ref{p-state}. As a sequential algorithm $\mathcal{A}(S)$ possesses a bounded exploration witness $W_S$, i.e. a finite set of terms such that $\Delta_{\mathcal{A}(S)}(S_1) = \Delta_{\mathcal{A}(S)}(S_2)$ holds, whenever states $S_1$ and $S_2$ coincide on $W_S$. We can always assume that $W_S$ just contains terms that must be evaluated in a state to determine the update set in that state. Thus, though $W_S$ is not unique we may assume that $W_S$ is somehow contained in the finite representation of $\mathcal{A}(S)$. This implies that the terms in $W_S$ result by interpretation from terms that appear in this representation, i.e. $W_S$ can be obtained using the extraction function $\beta$ that exists by the background postulate \ref{p-background}. Consequently, there must exist a finite set of terms $W$ such that its interpretation in a state yields both values and terms, and the latter ones represent $W_S$. We will continue to call $W$ a {\em bounded exploration witness}. Then the interpretation of $W$ and the interpretation of the extracted terms in any state suffice to determine the update set in that state. This leads to our {\em bounded exploration postulate} for RSAs. 

\begin{definition}\label{def-coincidence}\rm

Let $S$ and $S^\prime$ be states of an RSA, and let $W$ be a set of ground terms. We say that $S$ and $S^\prime$ {\em strongly coincide} over $W$ iff the following holds:

\begin{itemize}

\item For every $t \in W$ we have $\text{val}_S(t) = \text{val}_{S^\prime}(t)$.   

\item For every $t \in W$ with $\text{val}_S(t) \in \mathbb{T}_S$ and $\text{val}_{S^\prime}(t) \in \mathbb{T}_{S^\prime}$ we have \\
$\text{val}_S(\beta(t)) = \text{val}_{S^\prime}(\beta(t))$.
 
\end{itemize}

\end{definition}

Furthermore, we have to take a subtlety into account, which was not needed to be handled in the sequential ASM thesis. According to the background postulate the background structure must provide constructors and operators that permit the representation of sequential algorithms. While this leaves open, how such algorithms shall be represented by values, we can assume that bulk structures will be required. For instance, for rsASMs as defined in Section \ref{sec:rasm} we exploited trees, but according to our motivation at the beginning of Section \ref{sec:rsa} we could have used sets and tuples as well. We may further assume that the bulk values representing an algorithm are updated by several operations in one step, i.e. shared updates as defined in \cite{schewe:ejc2011} defined by an operator and arguments may be used to define updates. If operators are compatible\footnote{Conditions for the compatibility of shared updates have been discussed intensively in \cite{schewe:ejc2011}.}, such shared updates are merged into a single update of a bulk value. In order to capture this merging we extend the bounded exploration witness $W$ as follows.

A term indicating a shared update takes the form $op(f(t_1,\dots,t_n),t_1^\prime,\dots,t_m^\prime)$. Here $op$ is the operator that is to be applied, $f(t_1,\dots,t_n)$ evaluates in every state $S$ to a value $\text{val}_S(\ell)$ of some location $\ell = (f, (\text{val}_S(t_1) ,\dots, \text{val}_S(t_n)))$, and $t_1^\prime,\dots,t_m^\prime$ evaluate to the other arguments of the shared update. 

\begin{definition}\label{def-aggregation}\rm

If $op_1(f(t_1,\dots,t_n),t_{11}^\prime,\dots,t_{1m_1}^\prime) ,\dots$, $op_k(f(t_1,\dots,t_n),t_{k1}^\prime,\dots,t_{km_k}^\prime)$ are several terms of this form occurring in $W$ or in $\beta(W)$, then the term
\[ \hat{t} = op_1(op_2(\dots (op_k(f(t_1,\dots,t_n),t_{k1}^\prime,\dots,t_{km_k}^\prime)),\dots, 
t_{21}^\prime,\dots,t_{2m_2}^\prime), t_{11}^\prime,\dots,t_{1m_1}^\prime) \]
will be called an {\em aggregation term over $f(t_1,\dots,t_n)$}, and the tuple $(\text{val}_S(\hat{t})$, $ \text{val}_S(t_1) ,\dots, \text{val}_S(t_n))$ will be called an {\em aggregation tuple}. 

\end{definition}

Then we can always assume that the update set $\Delta_{\mathcal{A}}(S)$ is the result of collapsing an update multiset $\ddot{\Delta}_{\mathcal{A}}(S)$.

\begin{postulate}[Bounded Exploration Postulate]\label{p-bounded}
For every RSA $\mathcal{A}$ there is a finite set $W$ of ground terms such that $\ddot{\Delta}_{\mathcal{A}}(S) = \ddot{\Delta}_{\mathcal{A}}(S^\prime)$ holds (and consequently also $\Delta_{\mathcal{A}}(S) = \Delta_{\mathcal{A}}(S^\prime)$) whenever the states $S$ and $S^\prime$ strongly coincide over $W$. Furthermore, $\ddot{\Delta}_{\mathcal{A}}(\text{res}(S,\Sigma_{alg})) = \ddot{\Delta}_{\mathcal{A}}(\text{res}(S^\prime,\Sigma_{alg}))$ holds (and consequently also $\Delta_{\mathcal{A}}(\text{res}(S,\Sigma_{alg})) = \Delta_{\mathcal{A}}(\text{res}(S^\prime,\Sigma_{alg}))$) whenever the states $S$ and $S^\prime$ coincide over $W$. Here, $\text{res}(S,\Sigma_{alg})$ is the structure resulting from $S$ by restriction of the signature to $\Sigma_{alg}$.

\end{postulate}

Any set $W$ of ground terms as in the bounded exploration postulate \ref{p-bounded} will be called a {\em (reflective) bounded exploration witness} (R-witness) for $\mathcal{A}$. The four postulates capturing sequential time, abstract states, background and bounded exploration together provide an language-independent axiomatisation of the notion of a reflective sequential algorithm.

\begin{definition}\label{def-rsa}\rm

A {\em reflective sequential algorithm} (RSA) is defined by the sequential time postulate \ref{p-time}, the abstract state postulate \ref{p-state}, the background postulate \ref{p-background}, and the bounded exploration postulate \ref{p-bounded}.

\end{definition}

\subsection{Behavioural Equivalence}

According to Gurevich's definition in \cite{gurevich:tocl2000} two sequential algorithms are {\em behaviourally equivalent} iff they have the same sets of states and initial states and the same transition function $\tau$. Consequently, behaviourally equivalent sequential algorithms have the same runs. This may be weakened requesting directly that two sequential algorithms are behaviourally equivalent iff they have the same runs, e.g. one of the algorithms may have additional states that never appear in a run.

While for sequential algorithms this distinction is a mere subtlety without further consequences, it gains importance for RSAs. If we adopted without change the definition of behavioural equivalence from \cite{gurevich:tocl2000}, then the ``self-representation'', i.e. the substructure over $\Sigma_{alg}$ would be required to be exactly the same in corresponding states. However, the way how to realise such a representation was deliberately left open in the axiomatisation. For instance, rsASMs use a single $0$-ary function symbol and sophisticated tree structures, whereas at the beginning of this section we outlined that it is likewise possible to use sets of pairs comprising a Boolean condition and an update set. Therefore, instead of claiming identical states it suffices to only require identity for the restriction to $\Sigma_S - \Sigma_{alg}$, while structures over $\Sigma_{alg}$ only need to define behaviourally equivalent algorithms. However, this is still too restrictive, as behavioural equivalence of $\mathcal{A}(S)$ and $\mathcal{A}(S^\prime)$ (even, if these are considered merely as sequential, non-reflective algorithms) would still imply identical changes to the self-representation.

Therefore, we can also restrict our attention to the behaviour of the algorithms $\mathcal{A}(S)$ on states over $\Sigma_S - \Sigma_{alg}$. If $\mathcal{A}(S)$ and $\mathcal{A}(S^\prime)$ produce significantly different changes to the represented algorithm, the next state in a run of the reflective algorithm will reveal this. This leads to the following definition of behavioural equivalence for RSAs.

\begin{definition}\label{def-equivalence}\rm

Two RSAs $\mathcal{A}$ and $\mathcal{A}^\prime$ are {\em behaviourally equivalent} iff there exists a bijection $\Phi$ between runs of $\mathcal{A}$ and those of $\mathcal{A}^\prime$ such that for every run $S_0, S_1, \dots$ of $\mathcal{A}$ we have that for all states $S_i$ and $\Phi(S_i)$ 

\begin{enumerate}

\item their restrictions to $\Sigma_{S_i} - \Sigma_{alg}$ coincide, and

\item the restrictions of $\mathcal{A}(S_i)$ and $\mathcal{A}^\prime(\Phi(S_i))$ to $\Sigma_{S_i} - \Sigma_{alg}$ are behaviourally equivalent sequential algorithms.

\end{enumerate}

\end{definition}

Note that in our preliminary work in \cite{ferrarotti:psi2017} we did not request that states are closed under isomorphisms. Instead we permitted that the representations of sequential algorithms restricted to $\Sigma_S - \Sigma_{alg}$ are just behaviourally equivalent in states that are otherwise isomorphic. This would permit significantly different updates to the self representation to be subsumed in a single reflective algorithm. However, for the behavioural theory this is not plausible. We need Definition \ref{def-equivalence} in the construction of a behaviourally equivalent rsASM for an arbitrary RSA, for which no assumption on a specific self representation can be made.

\subsection{RSAs Defined by rsASMs}

We now show that rsASM satisfy the defining postulates, i.e. each rsASM defines an RSA. This constitutes the plausibility part of our reflective sequential ASM thesis.

\begin{theorem}\label{thm:plausibility}

Every reflective ASM $\mathcal{M}$ is a RSA.

\end{theorem}

\begin{proof}
First consider sequential time. According to Definition \ref{def-rsasm} the rsASM $\mathcal{M}$ comprises a set $\mathcal{I}$ of initial states defined over an (initial) signature $\Sigma$. Other states are defined through reachability by finitely many applications of the state transition function $\tau$. This gives rise to the set of states $\mathcal{S}$, where each state is defined over a signature $\Sigma_S$ with $\Sigma \subseteq \Sigma_S$. This state transition function $\tau$ is explicitly defined in Definition \ref{def-rsasm}. Furthermore, in every initial state $S_0 \in \mathcal{I}$ we have a unique rule $r_{S_0} = \textit{raise\/}(\textit{rule\/}(\text{val}_{S_0}(\textit{self\/})))$ using the rule extraction function $\textit{rule\/}$ defined in Subsection \ref{ssec:treerep}.

Concerning the abstract state postulate the required properties (1) and (5) are explicitly built into Definition \ref{def-rsasm} together with Definition \ref{def-extended-baseset}. Definition \ref{def-rsasm} further contains that the set $\mathcal{I}$ is closed under isomorphism. This extends to the set $\mathcal{S}$ of all states due to the definition of the state transition function $\tau$. This gives the required properties (2) and (3). Concerning property (4) we have $\Sigma_{alg} = \{ \textit{self\/} \}$, so the restriction of a state $S$ is simply given by $\text{val}_S(\textit{self\/})$. Applying the functions \textit{rule\/} and \textit{signature\/} from Subsection \ref{ssec:treerep} to this tree value give the rule $r_S$ and the signature $\Sigma_S$, which define the algorithm $\mathcal{A}(S)$ with the desired property.

The requirements for the background postulate are built into Definitions \ref{def-rsasm} and \ref{def-rsasm-background}. The extraction function $\beta$ has been defined explicitly in Subsection \ref{ssec:treerep}.

Finally, concerning the bounded exploration postulate take $W = \{ \textit{self\/} \}$. Let $S$ and $S^\prime$ be two states that strongly coincide on $W$. Then according to Definition \ref{def-coincidence} we have $\text{val}_S(\textit{self\/}) = \text{val}_{S^\prime}(\textit{self\/})$. As we have $r_S = \textit{raise\/}(\textit{rule\/}(\text{val}_S(\textit{self\/})))$, we obtain $r_S = r_{S^\prime}$ with signature $\Sigma_S = \textit{raise\/}(\textit{signature\/}(\text{val}_S(\textit{self\/}))) = \Sigma_{S^\prime}$. Furthermore, applying the extraction function $\beta$ gives $\beta(\text{val}_S(\textit{self\/})) = \beta(\text{val}_{S^\prime}(\textit{self\/}))$. Let this tuple of terms be $(t_1 ,\dots, t_n)$. Then the $t_i$ are all the terms over $\Sigma_S$ and the base set $B$ that appear in $r_S$ (and $r_{S^\prime}$). In particular, $\{ t_1 ,\dots, t_n \}$ is a bounded exploration witness for the sequential ASM defined by $\Sigma_S$ and $r_S$.

Definition \ref{def-coincidence} further implies $\text{val}_S(\textit{raise\/}(\beta(\textit{self\/}))) = \text{val}_{S^\prime}(\textit{raise\/}(\beta(\textit{self\/})))$, hence $\text{val}_S(\textit{raise\/}(t_i)) = \text{val}_{S^\prime}(\textit{raise\/}(t_i))$ holds for all $i = 1 ,\dots, n$, i.e. the states $S$ and $S^\prime$ coincide on a bounded exploration witness. Thus, we get $\ddot{\Delta}_{\mathcal{A}}(S) = \ddot{\Delta}_{r_S}(S) = \ddot{\Delta}_{r_{S^\prime}}(S^\prime) = \ddot{\Delta}_{\mathcal{A}}(S^\prime)$, which shows the satisfaction of the bounded exploration postulate and completes the proof.
\end{proof}

\section{The Reflective Sequential ASM Thesis}\label{sec:theory}

This section provides the mathematical proofs that rsASMs capture all RSAs, regardless how algorithms are represented by terms. We show the converse of Theorem \ref{thm:plausibility}, i.e. that every RSA $\mathcal{A}$ can be step-by-step simulated by a behaviourally equivalent rsASM $\mathcal{M}$.

\subsection{Critical Values}

Let $\mathcal{A}$ be given, and fix a bounded exploration witness $W$. Let $W_{st}$ be the subset of $W$ containing ``standard'' terms, i.e. those terms that do not contain function symbols in $\Sigma_{alg}$, and let $W_{pt} = W - W_{st}$ be the complement containing ``program terms''. 

According to the bounded exploration postulate the update set $\Delta_{\mathcal{A}}(S)$ in a state $S$ is determined by values $\text{val}_S(t)$ resulting from the interpretation of terms $t \in W$ and $t = \pi_i(\beta(\text{val}_S(t^\prime)))$ with $t^\prime \in W_{pt}$ and some projection function $\pi_i$. However, according to the abstract state postulate the successor state $\tau(S)$ is already determined by the sequential algorithm $\mathcal{A}(S)$ represented in the state $S$, i.e. $\Delta_{\mathcal{A}}(S) = \Delta_{\mathcal{A}(S)}(S)$ and therefore, instead of $\text{val}_S(t)$ with $t \in W$ it suffices to consider $t \in W_{st}$---values $\text{val}_S(t)$ with $t \in W_{pt}$ must already be covered by the values $\text{val}_S(\pi_i(\beta(\text{val}_S(t^\prime))))$ with $t^\prime \in W_{pt}$. We use the notation $W_{\beta} = \{ \beta(\text{val}_S(t)) \mid t \in W_{pt} \}$ and $W_S = W_{st} \cup W_{\beta}$. Without loss of generality we can assume that $W_S$ is closed under subterms.

\begin{definition}\label{def-critical}\rm 

For a state $S$ the terms in $W_S$ are called {\em critical terms} in $S$. The value $\text{val}_S(t)$ of a critical term $t \in W_S$ is called a {\em critical value} in $S$.

\end{definition}

In the following we proceed analogously to the proof of the main theorem in \cite{gurevich:tocl2000} concerning the capture of sequential algorithms by sequential ASMs, i.e. we start from a single state $S$ and the update set $\Delta_{\mathcal{A}}(S)$ in that state. We can first show that all values appearing in updates in an update set $\Delta_{\mathcal{A}}(S)$ are critical in $S$. The proof of this lemma is analogous to the corresponding proof for sequential ASMs (see \cite[Lemma 6.2]{gurevich:tocl2000}), but it has to be extended to capture also updates that result from aggregation of shared updates.

\begin{lemma}\label{lemma:critical_terms}

If $((f, (v_1, \dots v_n), v_0))$ is an update in $\Delta_{\mathcal{A}}(S)$, then either $v_0, v_1, \dots, v_n$ are critical values in $S$ or $(v_0, v_1, \dots, v_n)$ is an aggregation tuple (as defined in Definition \ref{def-aggregation}) built from critical values.

\end{lemma}

\begin{proof}
The update $((f, (v_1, \dots v_n)), v_0)$ may be the result of merging several shared updates or not. In the latter case assume that one value $v_i$ is not critical. Then choose a structure $S_1$ by replacing $v_i$ by a fresh value $b$ without changing anything else. Then $S_1$ is isomorphic to $S$ and thus a state by the abstract state postulate.

Let $t \in W_S$ be a critical term. Then we must have $\text{val}_S(t) = \text{val}_{S_1}(t)$, so $S$ and $S_1$ coincide on $W_S$, so they strongly coincide on $W$. The bounded exploration postulate implies $\Delta_{\mathcal{A}}(S) = \Delta_{\mathcal{A}}(S_1)$ and hence $(f(v_1,\dots,v_n),v_0) \in \Delta_{\mathcal{A}}(S_1)$. However, $v_i$ does not appear in the structure $S_1$ and thus cannot appear in $\Delta_{\mathcal{A}}(S_1)$, which gives a contradiction.

In case the update $((f, (v_1, \dots v_n)), v_0)$ is the result of merging several shared updates we have shared updates $((f, (v_1, \dots v_n)), op_i, v_{i1}^\prime,\dots,v_{im_i}^\prime)$ ($i = 1,\dots,k$), and
\[ v_0 = op_1(op_2(\dots (op_k(f_S(v_1,\dots,v_n),v_{k1}^\prime,\dots,v_{km_k}^\prime)),\dots, 
v_{21}^\prime,\dots,v_{2m_2}^\prime), v_{11}^\prime,\dots,v_{1m_1}^\prime) \]

Applying the same argument as in the first case (using $\ddot{\Delta}_{\mathcal{A}}(S) = \ddot{\Delta}_{\mathcal{A}}(S_1)$) we conclude that $(v_0, v_1, \dots, v_n)$ is an aggregation tuple.
\end{proof}

Lemma \ref{lemma:critical_terms} implies that every update in $\Delta_{\mathcal{A}}(S)$ can be produced by an ASM rule, which is either an assignment or the parallel composition of partial updates. This can be generalised to obtain a rule $r_S$ producing the whole update set $\Delta_{\mathcal{A}}(S)$. The lemma further implies that update sets (and update multisets) in a state $S$ are always finite.

\begin{lemma}\label{coro:rule}

For every state $S$ of the RSA $\mathcal{A}$ there is a rule $r_S$ using only terms in $W_S = W_{st} \cup W_{\beta}$ such that $\Delta_{r_S}(S) = \Delta_{\mathcal{A}}(S)$ holds. 

\end{lemma}

\begin{proof}
If $((f, (v_1, \dots v_n)), v_0) \in \Delta_{\mathcal{A}}(S)$ is not the result of merging several shared updates, then take critical terms $t_i \in W_S$ with $\text{val}_S(t_i) = v_i$ that are guaranteed by Lemma \ref{lemma:critical_terms}. Then the update is produced by the assignment rule $f(t_1,\dots,t_n) := t_0$.

If $((f, (v_1, \dots v_n)), v_0) \in \Delta_{\mathcal{A}}(S)$ is the result of merging several shared updates, then take critical terms $t_1, \dots, t_n$ and $t_{ij} \in W_S$ with $\text{val}_S(t_{ij}) = v_{ij}$ and operators $op_i$ that are guaranteed by Lemma \ref{lemma:critical_terms}. Then the update is produced by a parallel combination of partial assignment rules:
\[ \texttt{PAR} f(t_1,\dots,t_n) \leftleftarrows^{op_1} t_{11}^\prime ,\dots, t_{1m_1}^\prime \dots f(t_1,\dots,t_n) \leftleftarrows^{op_k} t_{k1}^\prime ,\dots, t_{km_k}^\prime \texttt{ENDPAR} \]

The parallel combination of all these rules for the individual updates in $\Delta_{\mathcal{A}}(S)$ gives the rule $r_S$. As $\Delta_{\mathcal{A}}(S)$ and $W_S$ are finite, this rule is well-defined.
\end{proof}

\subsection{Relative $W$-Similarity}

As in the corresponding proof of the sequential ASM thesis the following lemmata aim first to extend Lemma \ref{coro:rule} to other states $S^\prime$, i.e. to obtain $\Delta_{r_S}(S^\prime) = \Delta_{\mathcal{A}}(S^\prime)$, and to combine different rules $r_S$ such that the behaviour of $\mathcal{A}$ can be modelled by a combination of such rules on all states.

\begin{lemma}\label{lemma:coincide}

If two states $S$ and $S^\prime$ of $\mathcal{A}$ strongly coincide on $W$, then $\Delta_{r_S}(S^\prime) = \Delta_{\mathcal{A}}(S^\prime)$ holds.

\end{lemma}

\begin{proof}
As $S$ and $S^\prime$ strongly coincide on $W$ we also have $W_S = W_{S^\prime}$, and furthermore, $S$ and $S^\prime$ coincide on $W_S$. As the rule $r_S$ only uses terms in $W_S$, it follows that $\Delta_{r_S}(S) = \Delta_{r_S}(S^\prime)$ holds. Lemma \ref{coro:rule} also states $\Delta_{r_S}(S) = \Delta_{\mathcal{A}}(S)$, and the bounded exploration postulate gives $\Delta_{\mathcal{A}}(S) = \Delta_{\mathcal{A}}(S^\prime)$, which imply $\Delta_{r_S}(S^\prime) = \Delta_{\mathcal{A}}(S^\prime)$ as claimed.
\end{proof}

For the extensions to Lemma \ref{coro:rule} mentioned above we are naturally interested in states $S^\prime$, in which $\mathcal{A}(S^\prime)$ is behaviourally equivalent to $\mathcal{A}(S)$. For this we introduce the notion of $W S$-equivalence.

\begin{definition}\label{def-ws-equivalence}\rm

A state $S^\prime$ is {\em $W S$-equivalent} to the state $S$ iff $\beta(\text{val}_{S^\prime}(t)) =  \beta(\text{val}_S(t))$ holds for all $t \in W_{pt}$.

\end{definition}

Keeping in mind that the ``self-representation'', i.e. the restriction of a state $S$ to $\Sigma_{alg}$ that is to represent a sequential algorithm $\mathcal{A}(S)$, is de facto only relevant for the behaviour on ``standard'' locations, i.e. we are interested in the restrictions $\text{res}(S,\Sigma_S - \Sigma_{alg})$ of states---this is already reflected in Definition \ref{def-equivalence} concerning behavioural equivalence. 

We therefore use the notation $\Delta^{st}_{\mathcal{A}}(S) = res(\Delta_{\mathcal{A}}(S),\Sigma_S - \Sigma_{alg})$ and similarly $\Delta^{st}_{r_S}(S) = res(\Delta_{r_S}(S),\Sigma_S - \Sigma_{alg})$. Then the following lemma is a straightforward consequence of Lemma \ref{lemma:coincide}. 

\begin{lemma}\label{lem-corollary2}

If the states $S$ and $S^\prime$ are $W S$-equivalent and coincide over $W_{st} \cup W_\beta$, then we have $\Delta^{st}_{r_S}(S^\prime) = \Delta^{st}_{\mathcal{A}}(S^\prime)$.  

\end{lemma}

\begin{proof}
$W S$-equivalence implies that $W_{S^\prime} = W_S$. As $S$ and $S^\prime$ coincide on $W_S$, they strongly coincide on $W$, which gives $\Delta_{r_S}(S^\prime) = \Delta_{\mathcal{A}}(S^\prime)$ by Lemma \ref{lemma:coincide} and thus the claimed equality of the restricted update sets.
\end{proof}

\begin{definition}\label{def-relative-equivalence}\rm

Let $\mathcal{C}$ be a class of states. Two states $S_1, S_2 \in \mathcal{C}$ are called {\em $W$-similar relative to $\mathcal{C}$} iff $\sim_{S_1} = \sim_{S_2}$, where the equivalence relation $\sim_{S_i}$ on $W$ is defined by $t \sim_{S_i} t^\prime$ iff $\text{val}_{S_i}(t) = \text{val}_{S_i}(t^\prime)$.

\end{definition}

Naturally, we are mainly interested in classes $\mathcal{C}$ that are defined by $W S$-equivalence. We use the notation $[S]$ for the $W S$-equivalence class of the state $S$, i.e. $[S] = \{ S^\prime \mid S^\prime \;\text{is $W S$-equivalent to}\; S \}$. The following two lemmata extend Lemma \ref{coro:rule} to relative $W$-similar states.

\begin{lemma}\label{lem-isomorphism}

If states $S_1, S_2$ are isomorphic, and for state $S$ we have $\Delta^{st}_{r_S}(S_2) = \Delta^{st}_{\mathcal{A}}(S_2)$, then we also get $\Delta^{st}_{r_S}(S_1) = \Delta^{st}_{\mathcal{A}}(S_1)$.

\end{lemma}

\begin{proof}
Let $\sigma$ denote the isomorphism from $S_1$ to $S_2$, i.e. $S_2 = \sigma S_1$. Then $\Delta^{st}_{r_S}(S_2) = \sigma \Delta^{st}_{r_S}(S_1)$ and likewise $\Delta^{st}_{\mathcal{A}}(S_2) = \sigma \Delta^{st}_{\mathcal{A}}(S_1)$. This implies $\sigma \Delta^{st}_{r_S}(S_1) = \sigma \Delta^{st}_{\mathcal{A}}(S_1)$ and further $\Delta^{st}_{r_S}(S_1) = \Delta^{st}_{\mathcal{A}}(S_1)$ by applying $\sigma^{-1}$ to both sides.
\end{proof}

\begin{lemma}\label{lem-relative-equivalence}

If states $S_1$ and $S_2$ are $W$-similar relative to $[S]$, then $\Delta^{st}_{r_{S_1}}(S_2) = \Delta^{st}_{\mathcal{A}}(S_2)$. 

\end{lemma}

\begin{proof}
If we replace every element in the base set of $S_2$ that also belongs to the base set of $S_1$ by a fresh element, we obtain a structure $S^\prime$ isomorphic to $S_2$ and disjoint from $S_1$. By the abstract state postulate $S^\prime$ is a state of $\mathcal{A}$. Furthermore, by construction $S^\prime$ and $S_2$ are also $W$-similar relative to $[S]$. Due to Lemma \ref{lem-isomorphism} it suffices to show $\Delta^{st}_{r_{S_1}}(S^\prime) = \Delta^{st}_{\mathcal{A}}(S^\prime)$, so without loss of generality we may assume that the base sets of $S_1$ and $S_2$ are disjoint. 

Define a new structure $\hat{S}$ by replacing in $S_2$ all values $\text{val}_{S_2}(t)$ with a critical term $t \in W_S$ by the corresponding value $\text{val}_{S_1}(t)$. As $S_2 \in [S]$ holds, we have $W_{S_2} = W_S$, and due to $W$-similarity relative to $[S]$ we have $\text{val}_{S_1}(t) = \text{val}_{S_1}(t^\prime)$ iff $\text{val}_{S_2}(t) = \text{val}_{S_2}(t^\prime)$ holds, so the structure $\hat{S}$ is well-defined. Furthermore, $\hat{S}$ is isomorphic to $S_2$ and thus a state by the abstract state postulate. We get $\hat{S} \in [S]$ by construction. Furthermore, $S_1$ and $\hat{S}$ coincide on $W_S$, so Lemma \ref{lem-corollary2} implies that $\Delta^{st}_{r_{S_1}}(\hat{S}) = \Delta^{st}_{\mathcal{A}}(\hat{S})$ holds. Using again Lemma \ref{lem-isomorphism} completes the proof.
\end{proof}

Next we define a rule $r_{[S]}$ for a whole $W S$-equivalence class $[S]$. For $S_1 \in [S]$ let $\varphi_{S_1}$ be the following Boolean term:
\[ \bigwedge_{\substack{t_i, t_j \in W_{st} \,\cup\, W_{\beta} \\ \text{val}_{S_1}(t_i) = \text{val}_{S_1}(t_j)}} t_i = t_j \quad \wedge \bigwedge_{\substack{t_i, t_j \in W_{st} \,\cup\, W_{\beta} \\ \text{val}_{S_1}(t_i) \neq \text{val}_{S_1}(t_j)}} \neg (t_i = t_j). \] 

Clearly, a state $S_2 \in [S]$ satisfies $\varphi_{S_1}$ iff $S_1$ and $S_2$ are $W$-similar relative to $[S]$. As $W$ is finite, we obtain a partition of $[S]$ into classes $[S]_1 ,\dots, [S]_n$ such that two states belong to the same class $[S]_i$ iff they are $W$-similar relative to $[S]$. We choose representatives $S_1 ,\dots, S_n$ for these classes and define the rule $r_{[S]}$ by
\[ \texttt{PAR}\; (\texttt{IF}\; \varphi_{S_1} \;\texttt{THEN}\; r_{S_1} \;\texttt{ENDIF}) \; \dots \;
(\texttt{IF}\; \varphi_{S_n} \;\texttt{THEN}\; r_{S_n} \;\texttt{ENDIF}) \; \texttt{ENDPAR} \]

Using this rule we obtain the following lemma, which is a straightforward consequence of the previous lemmata. 

\begin{lemma}\label{lem-relatBehavEquiv}  

For every state $S$ and every state $S^\prime \in [S]$ we have $\Delta^{st}_{r_{[S]}}(S^\prime) = \Delta^{st}_{\mathcal{A}}(S^\prime)$.

\end{lemma} 

\begin{proof}
There is exactly one class $[S]_i$ with representing state $S_i$ such that $S^\prime \in [S]_i$ holds. Then $\text{val}_{S^\prime}(\varphi_j)$ is \textbf{true} iff $j = i$. Then we get $\Delta^{st}_{r_[S]}(S^\prime) = \Delta^{st}_{r_{S_i}}(S^\prime) = \Delta^{st}_{\mathcal{A}}(S^\prime)$ using the definition of $r_{[S]}$ and Lemma \ref{lem-relative-equivalence}.
\end{proof}

\subsection{Tree Updates}

We need to combine the partial results obtained so far into the construction of a rsASM $\mathcal{M}$ that is behaviourally equivalent to the given RSA $\mathcal{A}$. Lemma \ref{coro:rule} shows that each state transition can be expressed by a rule, i.e. for each state $S$ we have a rule $r_S$ with $\Delta_{r_S}(S) = \Delta_{\mathcal{A}}(S)$, hence $\tau(S) = S + \Delta_{r_S}(S)$. Lemma \ref{lem-relatBehavEquiv} gives a single rule $r_{[S]}$ for each class $[S]$ with $\Delta^{st}_{r_{[S]}}(S^\prime) = \Delta^{st}_{\mathcal{A}}(S^\prime)$ for all $S^\prime \in [S]$, hence $res(\tau(S^\prime), \Sigma_{\tau(S^\prime)} - \Sigma_{alg}) = res(S^\prime + \Delta_{r_{[S]}}(S^\prime), \Sigma_{\tau(S^\prime)} - \Sigma_{alg})$. The latter result already captures the updates on the ``standard'' part of the state by an ASM rule, but the whole rule $r_{[S]}$ updates also locations with function symbols in $\Sigma_{alg}$ instead of \textit{self\/}.

What we need is a different representation of a rule in a state extending the restriction to $\Sigma_S - \Sigma_{alg}$ and using a location \textit{self\/}. For any state $S$ of the RSA $\mathcal{A}$ let $\Phi(S)$ denote a structure over $(\Sigma_S - \Sigma_{alg}) \cup \{ \textit{self\/} \}$, in which the location \textit{self\/} captures this intended representation. The structure $\Phi(S)$ should be built as follows:

\begin{enumerate}

\item We must have $res(\Phi(S), \Sigma_S - \Sigma_{alg}) = res(S, \Sigma_S - \Sigma_{alg})$.

\item The value $\text{val}_S(\textit{self})$ must be a tree composing a signature subtree and a rule subtree, i.e. $\text{val}_S(\textit{self}) =$
\[ \textit{label\_hedge\/}(\texttt{self}, \textit{label\_hedge\/}(\texttt{signature}, s_0 s_1 \dots s_k), \textit{label\_hedge\/}(\texttt{rule}, r) . \]

\item The signature subtree simply lists the signature $\Sigma_{\Phi(S)} = (\Sigma_S - \Sigma_{alg}) \cup \{ \textit{self\/} \}$, so with $f_0 = \textit{self\/}$, $\{ f_1 ,\dots, f_k \} = \Sigma_S - \Sigma_{alg}$, $n_0 = 0$, and $ar(f_i) = n_i$ for all $i$ each of the subtrees $s_i$ takes the form
\[ \textit{label\_hedge\/}(\texttt{func}, \texttt{name} \langle f_i \rangle \; \texttt{arity} \langle n_i \rangle) . \]

\item The rule subtree captures a representation of a rule $r_{\Phi(S)}$, which must be a parallel composition of two parts: (a) a rule capturing the updates to $res(\Phi(S), \Sigma_S - \Sigma_{alg})$, (b) a rule capturing updates to \textit{self\/}.

\end{enumerate}

Concerning the updates to $res(\Phi(S), \Sigma_S - \Sigma_{alg})$ we know from Lemma \ref{lem-relatBehavEquiv} that these are captured by the rule $r_{[S]}$ restricted to $\Sigma_S - \Sigma_{alg}$, so the representation as a tree is done in the standard way. Let $t_{r_{[S]}}$ denote this tree.

The remaining problem is to find a rule capturing the updates to the tree assigned to \textit{self\/} in any state $\Phi(S)$. For this we will proceed analogously to Lemmata \ref{coro:rule}-\ref{lem-relatBehavEquiv} focusing on finitely many changes to trees. In particular, it will turn out that a fixed tree representing a rule $r^{self}$ will be sufficient for all states $S$.

We start with a tree $\hat{t}_S$ satisfying the requirements from (2), (3) and (4)(a) above, but ignoring updates to \textit{self\/}. This tree is defined as 
\[ \textit{label\_hedge\/}(\texttt{self}, \textit{label\_hedge\/}(\texttt{signature}, s_0 s_1 \dots s_k), \textit{label\_hedge\/}(\texttt{rule}, t_{r_{[S]}}) . \]

Furthermore, we let $\hat{\Phi}(S)$ denote the structure defined by $res(\hat{\Phi}(S), \Sigma_S - \Sigma_{alg}) = res(S, \Sigma_S - \Sigma_{alg})$ and $\text{val}_{\hat{\Phi}(S)}(\textit{self\/}) = \hat{t}_S$.

\begin{lemma}\label{lem-tree-update}

There exists a rule $r_S^{self}$ with $\Delta_{r_S^{self}}(\hat{\Phi}(S)) = \{ (\textit{self\/}, \hat{t}_{\tau(S)} ) \}$. The rule uses only the operators of the tree algebra and terms of the form $\textit{raise\/}(t)$ with $t$ appearing in $\text{val}_{\hat{\Phi}(S)}(\textit{self\/})$.

\end{lemma}

\begin{proof}
Concerning the signature finitely many new function symbols are added to $\Sigma_S$ by the transition of $\mathcal{A}$ to $\tau(S)$. So we can define a rule $r_{\hat{\Phi}(S)}^{sig}$ as a finite parallel composition of rules $\textit{NewFunc\/}(f_i,n_i)$ ($i = 1,\dots,k$), where $\textit{NewFunc\/}(f,n)$ is defined as\footnote{Note that \textit{sign\/} is a logical variable bound to a node in the tree. As such it cannot appear on the left-hand side of an assignment or a partial assignment, but the sublocation corresponding to this node can. This sublocation is \textit{raise\/}(\textit{sign\/}), and we wrote in Section \ref{sec:rasm} that we will blur the syntactic distinction between values denoting function symbols and the function symbols themselves, as there is no risk of confusion. The collapse of the update multiset containing the shared updates defined in this way give rise to the explicit construction as used in Example \ref{bsp-tree}.}
\begin{gather*}
\texttt{LET}\; \textit{sign\/}\; = \textbf{I} o . ( \textit{root\/}(\textit{self\/}) \prec_c o \wedge \textit{label\/}(o) = \texttt{signature} ) \; \texttt{IN} \\
\textit{sign\/}\; \leftleftarrows^{\textit{right\_extend\/}}\; \textit{label\_hedge\/}(\texttt{func}, \langle f \rangle \; \langle n \rangle ) 
\end{gather*}

According to Lemma \ref{coro:rule} we have $\tau(S) = S + \Delta_{r_S}(S)$, so every new function symbol $f_i$ must be taken from the reserve $\Sigma_{res}$ and its arity $n_i$ must be the value of a term used in $r_S$. Such terms also occur in $r_{[S]}$, hence they must have the form $\textit{raise\/}(t)$ with $t$ appearing in $\text{val}_{\hat{\Phi}(S)}(\textit{self\/})$.

Concerning the subtrees representing the rule represented in \textit{self\/} the rules $r_{[S]}$ and $r_{[\tau(S)]}$ take the form
\[ \texttt{PAR}\; \texttt{IF}\; \varphi_1 \;\texttt{THEN}\; r_1 \;\texttt{ENDIF}\; \dots \texttt{IF}\; \varphi_l \;\texttt{THEN}\; r_l \;\texttt{ENDIF} \; \texttt{ENDPAR}\; , \]

\noindent
in which the rules $r_i$ are parallel compositions of assignment rules $f(t_1,\dots,t_n) := t_0$ with $f \in \Sigma_S - \Sigma_{alg}$ and partial update rules $f(t_1,\dots,t_n) \leftleftarrows^{op} t_1^\prime ,\dots, t_m^\prime$. These are represented by trees of the form
\[ \textit{label\_hedge}(\texttt{update},\texttt{func} \langle f \rangle \; \texttt{term} \langle t_1 \dots t_n \rangle \; \texttt{term} \langle t_0 \rangle) \]
and
\[ \textit{label\_hedge}(\texttt{partial},\texttt{func} \langle f \rangle \; \texttt{func} \langle op \rangle \; \texttt{term} \langle t_1 \dots t_n \rangle \; \texttt{term} \langle t_1^\prime \dots t_m^\prime \rangle) , \]

respectively. According to Proposition \ref{prop-tree-difference} $t_{r_{[\tau(S)]}}$ can be defined by a tree algebra expression on $t_{r_{[S]}}$ with values of the form $\textit{drop\/}(t)$ for $t$ appearing in $r_{[S]}$. Again, these terms must have the form $\textit{raise\/}(t)$ with $t$ appearing in $\text{val}_{\hat{\Phi}(S)}(\textit{self\/})$. Corollary \ref{prop-tree-update} defines a rule $r_{\hat{\Phi}(S)}^{rule}$, and $r_S^{self}$ is then the parallel composition of $r_{\hat{\Phi}(S)}^{sig}$ and $r_{\hat{\Phi}(S)}^{rule}$.
\end{proof}

Lemma \ref{lem-tree-update} shows that we can always find a rule that transforms one tree representation into another one. If we had chosen a different self-representation, we would need similarly powerful manipulation operators that ensure an analogous result. For our proof here it will be essential to show that only finitely many such rules are needed. We will show this by the following sequence of lemmata, by means of which we extend the applicability of the rule $r_S^{\textit{self\/}}$ to trees $\hat{\Phi}(S^\prime)$ for other states $S^\prime$.

\begin{lemma}\label{lem-tree-coincidence}

If states $S$ and $S^\prime$ coincide on $W$, then $\Delta_{r_S^{self}}(\hat{\Phi}(S^\prime)) = \{ (\textit{self\/}, \hat{t}_{\tau(S^\prime)} ) \}$.

\end{lemma}

\begin{proof}
The updates from $\mathcal{A}(S)$ to $\mathcal{A}(\tau(S))$ are defined by $\text{res}(\Delta_{\mathcal{A}}(S), \Sigma_{alg})$: we have $\text{res}(S, \Sigma_{alg}) + \text{res}(\Delta_{\mathcal{A}}(S), \Sigma_{alg}) = \text{res}(\tau(S), \Sigma_{alg})$. Analogously, the changes from $\mathcal{A}(S^\prime)$ to $\mathcal{A}(\tau(S^\prime))$ are defined by $\text{res}(\Delta_{\mathcal{A}}(S^\prime), \Sigma_{alg})$: we have $\text{res}(S^\prime, \Sigma_{alg}) + \text{res}(\Delta_{\mathcal{A}}(S^\prime), \Sigma_{alg}) = \text{res}(\tau(S^\prime), \Sigma_{alg})$. Then this also applies to the restrictions of $\mathcal{A}(S)$ and $\mathcal{A}(S^\prime)$ to the states defined over $\Sigma_S - \Sigma_{alg}$ and $\Sigma_{S^\prime} - \Sigma_{alg}$, respectively. These restricted algorithms are behaviourally equivalent to $r_{[S]}$ and $r_{[\tau(S)]}$.

Thus, the updates from $r_{[S^\prime]}$ to $r_{[\tau(S^\prime)]}$ are defined by $\text{res}(\Delta_{\mathcal{A}}(S^\prime), \Sigma_{alg})$. As $S$ and $S^\prime$ coincide on $W$ the bounded exploration postulate implies that $\text{res}(\Delta_{\mathcal{A}}(S^\prime), \Sigma_{alg}) = \text{res}(\Delta_{\mathcal{A}}(S), \Sigma_{alg})$ holds. As the updates from $r_{[S]}$ to $r_{[\tau(S)]}$ using their tree representations are equivalently expressed by $r_S^{\textit{self\/}}$, we conclude that the updates from the tree representation of $r_{[S^\prime]}$ to the tree representation of $r_{[\tau(S^\prime)]}$ are also defined by $r_S^{\textit{self\/}}$ as claimed.
\end{proof}

\begin{lemma}\label{lem-tree-isomorphism}

If states $S_1, S_2$ are isomorphic, and for a state $S$ we have $\Delta_{r_S^{\textit{self\/}}}(\hat{\Phi}(S_2)) = \{ (\textit{self\/}, \hat{t}_{S_2} ) \}$, then we also get $\Delta_{r_S^{\textit{self\/}}}(\hat{\Phi}(S_1)) = \{ (\textit{self\/}, \hat{t}_{S_1} ) \}$.

\end{lemma}

\begin{proof}
Let $\sigma$ denote the isomorphism from $S_1$ to $S_2$, i.e. $S_2 = \sigma S_1$. Then $\Delta_{r_S^{\textit{self\/}}}(\hat{\Phi}(S_2)) = \sigma \Delta_{r_S^{\textit{self\/}}}(\hat{\Phi}(S_1))$ and likewise $\hat{t}_{S_2} = \sigma \hat{t}_{S_1}$. 

This implies $\sigma \Delta_{r_S^{\textit{self\/}}}(\hat{\Phi}(S_1)) = \sigma \{ (\textit{self\/}, \hat{t}_{S_1}) \}$ and further $\Delta_{r_S^{\textit{self\/}}}(\hat{\Phi}(S_1)) = \{ (\textit{self\/}, \hat{t}_{S_1}) \}$ by applying $\sigma^{-1}$ to both sides.
\end{proof}

\subsection{$W$-Similarity}

For Lemma \ref{lem-relative-equivalence} we exploited relative $W$-similarity, which exploits terms in classes $[S]$ defined by the evaluation of the bounded exploration witness $W$. As the class $[S]$ depends on the evaluation of the terms in $W_{pt}$, also relative equivalence implicitly depends on $W_{pt}$. Now we require a notion of $W$-similarity that is grounded only on $W$.

\begin{definition}\label{def-wsimilarity}\rm

Two states $S_1, S_2$ are called {\em $W$-similar} iff $\sim_{S_1} = \sim_{S_2}$ holds, where the equivalence relation $\sim_{S_i}$ on $W$ is defined by $t \sim_{S_i} t^\prime$ iff $\text{val}_{S_i}(t) = \text{val}_{S_i}(t^\prime)$.

\end{definition}

The proof of the following lemma will be quite analogous to the proof of Lemma \ref{lem-relative-equivalence}.

\begin{lemma}\label{lem-tree-similarity}

If states $S$ and $S^\prime$ are $W$-similar, then $\Delta_{r_S^{self}}(\hat{\Phi}(S^\prime)) = \{ (\textit{self\/}, \hat{t}_{\tau(S^\prime)} ) \}$.

\end{lemma}

\begin{proof}
If we replace every element in the base set of $S^\prime$ that also belongs to the base set of $S$ by a fresh element, we obtain a structure $S^{\prime\prime}$ isomorphic to $S^\prime$ and disjoint from $S$. By the abstract state postulate $S^{\prime\prime}$ is also a state of $\mathcal{A}$. Furthermore, by construction $S^\prime$ and $S^{\prime\prime}$ are also $W$-similar to $S$. Lemma \ref{lem-isomorphism} already covers equality on the ``standard part'' of the signature, so it suffices to show $\Delta_{r_S^{self}}(\hat{\Phi}(S^{\prime\prime})) = \{ (\textit{self\/}, \hat{t}_{\tau(S^{\prime\prime})} ) \}$, so without loss of generality we may assume that the base sets of $S$ and $S^\prime$ are disjoint. 

Define a new structure $\hat{S}$ by replacing in $S^\prime$ all values $\text{val}_{S^\prime}(t)$ with a critical term $t \in W$ by the corresponding value $\text{val}_{S}(t)$. Due to $W$-similarity we have $\text{val}_{S}(t) = \text{val}_{S}(t^\prime)$ iff $\text{val}_{S^\prime}(t) = \text{val}_{S^\prime}(t^\prime)$ holds, so the structure $\hat{S}$ is well-defined. Furthermore, $\hat{S}$ is isomorphic to $S^\prime$ and thus a state by the abstract state postulate. Furthermore, $S$ and $\hat{S}$ coincide on $W$, so Lemma \ref{lem-tree-coincidence} implies that $\Delta_{r_S^{self}}(\hat{\Phi}(\hat{S})) = \{ (\textit{self\/}, \hat{t}_{\tau(\hat{S})} ) \}$ holds. Using again Lemma \ref{lem-tree-isomorphism} completes the proof.
\end{proof}

Using Lemma \ref{lem-tree-similarity} we can exploit that there are only finitely many $W$-similarity classes to define a single rule $r^{\textit{self\/}}$ for all states. Let $\varphi_S$ be the following Boolean term:
\[ \bigwedge_{\substack{t_i, t_j \in W \\ \text{val}_{S}(t_i) = \text{val}_{S}(t_j)}} t_i = t_j \quad \wedge \bigwedge_{\substack{t_i, t_j \in W \\ \text{val}_{S}(t_i) \neq \text{val}_{S}(t_j)}} \neg (t_i = t_j). \] 

Clearly, a state $S^\prime$ satisfies $\varphi_S$ iff $S$ and $S^\prime$ are $W$-similar. As $W$ is finite, we obtain a partition of $\mathcal{S}$ into classes $[S_1] ,\dots, [S_n]$ with representatives $S_i$ ($i=1,\dots,n$) such that a state $S$ belongs to the class $[S_i]$ iff $S$ and $S_i$ are $W$-similar. We define the rule $r^{\textit{self\/}}$ by
\[ \texttt{PAR}\; (\texttt{IF}\; \varphi_{S_1} \;\texttt{THEN}\; r_{S_1}^{\textit{self\/}} \;\texttt{ENDIF}) \; \dots \;
(\texttt{IF}\; \varphi_{S_n} \;\texttt{THEN}\; r_{S_n}^{\textit{self\/}} \;\texttt{ENDIF}) \; \texttt{ENDPAR} \]

Using this rule we obtain the following lemma, which is a straightforward consequence of the previous lemmata. 

\begin{lemma}\label{lem-treeBehavEquiv}  

For all states $S$ we have $\Delta_{r^{self}}(\hat{\Phi}(S)) = \{ (\textit{self\/}, \hat{t}_{\tau(S)} ) \}$.

\end{lemma} 

\begin{proof}
There is exactly one class $[S_i]$ with representing state $S_i$ such that $S \in [S_i]$ holds. Then $\text{val}_{S}(\varphi_j)$ is \textbf{true} iff $j = i$. Then we get $\Delta_{r^{self}}(\hat{\Phi}(S)) = \Delta_{r_{S_i}^{self}}(\hat{\Phi}(S)) = \{ (\textit{self\/}, \hat{t}_{\tau(S)} ) \}$ using the definition of $r^{\textit{self\/}}$ and Lemma \ref{lem-tree-similarity}.
\end{proof}

We now extend the tree $\hat{t}_S$ to a definition of a tree $t_S$ satisfying all the requirements from (2), (3) and (4) above including now the updates to \textit{self\/}. This tree is defined as 
\begin{gather*}
\textit{label\_hedge\/}(\texttt{self}, \textit{label\_hedge\/}(\texttt{signature}, s_0 s_1 \dots s_k), \\ \textit{label\_hedge\/}(\texttt{rule}, \textit{label\_hedge\/}(\texttt{par}, t_{r_{[S]}} \; t_{r^{\textit{self\/}}})) , 
\end{gather*}

where $t_{r^{\textit{self\/}}}$ is the tree representation of the rule $r^{\textit{self\/}}$. Furthermore, we let $\Phi(S)$ denote the structure defined by $res(\Phi(S), \Sigma_S - \Sigma_{alg}) = res(S, \Sigma_S - \Sigma_{alg})$ and $\text{val}_{\Phi(S)}(\textit{self\/}) = t_S$. We use this to define an rsASM $\mathcal{M}$.

\begin{lemma}\label{lem-rsasm}

For each RSA $\mathcal{A}$ we obtain an rsASM $\mathcal{M}$  with the set $\mathcal{S}_{\mathcal{M}} = \{ \Phi(S) \mid S \in \mathcal{S} \}$ of states, the set $\mathcal{I}_{\mathcal{M}} = \{ \Phi(S) \mid S \in \mathcal{I} \}$ of initial states, and the state transition function $\tau_{\mathcal{M}}$ with $\tau_{\mathcal{M}}(\Phi(S)) = \Phi(\tau(S))$.

\end{lemma}

\begin{proof}
According to the abstract state postulate all initial states $S_0, S_0^\prime \in \mathcal{I}$ are defined over the same signature $\Sigma_0 \cup \Sigma_{alg}$. As $\mathcal{I}$ is closed under isomorphisms, this also holds for $\mathcal{I}_{\mathcal{M}}$. As $\mathcal{A}(S_0) = \mathcal{A}(S_0^\prime)$ holds, the states $\Phi(S_0)$ and $\Phi(S_0^\prime)$ coincide on \textit{self\/}.

We have $\textit{raise\/}(\textit{rule\/}(\text{val}_{\Phi(S)}(\textit{self\/}))) = \texttt{PAR}\; r_{[S]} \; r^{\textit{self\/}} \; \texttt{ENDPAR}$. Denoting this rule as $r_{\Phi(S)}$ we obtain
\[ \Delta_{r_{\Phi(S)}}(\Phi(S)) = \Delta^{st}_{r_{[S]}}(S) \cup \{ (\textit{self\/}, t_{\tau(S)} ) \} , \]

which implies $\Phi(S) + \Delta_{r_{\Phi(S)}}(\Phi(S)) = \Phi(\tau(S))$ using Lemmata \ref{lem-relatBehavEquiv} and \ref{lem-treeBehavEquiv}. Thus, all requirements from Definition \ref{def-rsasm} are satisfied.
\end{proof}

With this construction of an rsASM $\mathcal{M}$ from an RSA $\mathcal{A}$ we can state and prove the main result of this section.

\begin{theorem}

For every RSA $\mathcal{A}$ there is a behaviourally equivalent rsASM $\mathcal{M}$.

\end{theorem}

\begin{proof}
Given the RSA $\mathcal{A}$ we define the rsASM $\mathcal{M}$ as in Lemma \ref{lem-rsasm}. It remains to show that $\mathcal{A}$ and $\mathcal{M}$ are behaviourally equivalent in the sense of Definition \ref{def-equivalence}.

For this let $S_0, S_1, \dots$ be an arbitrary run of $\mathcal{A}$. Then according to Lemma \ref{lem-rsasm} $\Phi(S_0), \Phi(S_1), \dots$ is a run of $\mathcal{M}$, which defines the required bijection between runs of $\mathcal{A}$ and $\mathcal{M}$. 

Due to our construction of the states $\Phi(S)$ we further have $\text{res}(\mathcal{A}, \Sigma_{S} - \Sigma_{alg}) = \text{res}(\mathcal{M}, \Sigma_{S} - \Sigma_{alg})$, which implies property (1) of Definition \ref{def-equivalence}.

The restriction of the sequential algorithm $\mathcal{A}(S)$ to $\Sigma_{S} - \Sigma_{alg}$ is expressed by the rule $r_{[S]}$. The sequential algorithm represented in $\Phi(S)$ is expressed by the rule $\texttt{PAR}\; r_{[S]} \; r^{\textit{self\/}} \;\texttt{ENDPAR}$, and its restriction to $\Sigma_{S} - \Sigma_{alg}$ is also expressed by the rule $r_{[S]}$, which implies property (2) of Definition \ref{def-equivalence}.
\end{proof}

\section{Related Work}\label{sec:others}

The ur-instance of a behavioural theory is Gurevich's celebrated {\em sequential ASM thesis} \cite{gurevich:tocl2000}, which states and proves that sequential algorithms are captured by sequential ASMs. A key contribution of this thesis is the language-independent definition of a sequential algorithm by a small set of intuitively understandable postulates on an arbitrary level of abstraction. As the sequential ASM thesis shows, the notion of sequential algorithm includes a form of bounded parallelism, which is a priori defined by the algorithm and does not depend on the actual state. However, parallel algorithms, e.g. for graph inversion or leader election, require unbounded parallelism. A behavioural theory of synchronous parallel algorithms has been first approached by Blass and Gurevich \cite{blass:tocl2003,blass:tocl2008}, but different from the sequential thesis the theory was not accepted, not even by the ASM community despite its inherent proof that ASMs \cite{boerger:2003} capture parallel algorithms. One reason is that the axiomatic definition exploits non-logical concepts such as mailbox, display and ken, whereas the sequential thesis only used logical concepts such as structures and sets of terms. Even the background, that is left implicit in the sequential thesis, only refers to truth values and operations on them.

In \cite{ferrarotti:tcs2016} an alternative behavioural theory of synchronous parallel algorithms (aka ``simplified parallel ASM thesis'') was developed. It was inspired by previous research on a behavioural theory for non-deterministic database transformations \cite{schewe:ac2010}. Largely following the careful motivation in \cite{blass:tocl2003} it was first conjectured in \cite{schewe:abz2012} that it should be sufficient to generalise bounded exploration witnesses to sets of multiset comprehension terms. The rationale behind this conjecture is that in a particular state the multiset comprehension terms give rise to multisets, and selecting one value out each of these multisets defines the proclets used by Blass and Gurevich. The formal proof of the simplified ASM thesis in \cite{ferrarotti:tcs2016} requires among others an investigation in finite model theory. At the same time another behavioural theory of parallel algorithms was developed in \cite{dershowitz:igpl2016}, which is independent from the simplified parallel ASM thesis, but refers also to previous work by Blass and Gurevich. A thorough comparison with the simplified parallel ASM thesis has not yet been conducted.

There have been many attempts to capture asynchronous parallelism as marked in theories of concurrency as well as distribution (see \cite{lynch:1996} for a collection of many distributed or concurrent algorithms). Gurevich's axiomatic definition of partially ordered runs \cite{gurevich:lipari1995} tries to reduce the problem to families of sequential algorithms, but the theory is too strict. As shown in \cite{boerger:ai2016} it is easy to find concurrent algorithms that satisfy sequential consistency \cite{lamport:tc1979}, where runs are not partially ordered. One problem is that the requirements for partially ordered runs always lead to linearisability. The lack of a convincing definition of asynchronous parallel algorithms was overcome by the work on concurrent algorithms in \cite{boerger:ai2016}, in which a concurrent algorithm is defined by a family of agents, each equipped with a sequential algorithm with shared locations. While each individual sequential algorithm in the family is defined by the postulates for sequential algorithms\footnote{A remark in \cite{boerger:ai2016} states that the restriction to sequential algorithms is not really needed. An extension to concurrent algorithms covering families of parallel algorithms is sketched in \cite{schewe:acsw2017}.}, the family as a whole is subject to a concurrency postulate requiring that a successor state of the global state of the concurrent algorithm results from simultaneously applying update sets of finitely many agents that have been built on some previous (not necessarily the latest) states. The theory shows that concurrent algorithms are captured by concurrent ASMs. As in concurrent algorithms, in particular in case of distribution, message passing between agents is more common than shared locations, it has further been shown in \cite{boerger:jucs2017} that message passing can be captured by regarding mailboxes as shared locations, which leads to communicating concurrent ASMs capturing concurrent algorithms with message passing.

\section{Conclusion}\label{sec:schluss}

In this article we investigated a behavioural theory for reflective sequential algorithms (RSAs) following our sketch in \cite{ferrarotti:psi2017}. Grounded in related work concerning behavioural theories for sequential algorithms \cite{gurevich:tocl2000}, (synchronous) parallel algorithms \cite{ferrarotti:tcs2016}, and concurrent algorithms \cite{boerger:ai2016} we developed a set of abstract postulates defining RSAs, extended ASMs to reflective sequential abstract state machines (rsASMs), and proved that any RSA as stipulated by the postulates can be step-by-step simulated by an rsASM. The key contributions are the axiomatic definition of RSAa and the proof that RSAs are captured by rsASMs.

With this behavioural theory we lay the foundations for rigorous development of reflective algorithms and thus adaptive systems. However, the theory in this article covers only reflective {\em sequential} algorithms, so in view of the behavioural theories for unbounded parallel and concurrent algorithms the next steps of the research are to extend these theories to capture also reflection. We envision a part II addressing reflective parallel algorithms and a part III on reflective concurrent algorithms. The latter one would then lay the foundations for the specification of distributed adaptive systems in general.

Furthermore, for rigorous development extensions to the refinement method for ASMs \cite{boerger:fac2003} and to the logic used for verification \cite{ferrarotti:igpl2017,ferrarotti:amai2018} will be necessary. These will also be addressed in follow-on research.

\bibliographystyle{abbrv}
\bibliography{rsa}

\begin{thebibliography}{10}

\bibitem{blass:tocl2003}
A.~Blass and Y.~Gurevich.
\newblock {Abstract State Machines} capture parallel algorithms.
\newblock {\em ACM Trans. Computational Logic}, 4(4):578--651, 2003.

\bibitem{blass:beatcs2007}
A.~Blass and Y.~Gurevich.
\newblock Background of computation.
\newblock {\em Bulletin of the {EATCS}}, 92:82--114, 2007.

\bibitem{blass:tocl2008}
A.~Blass and Y.~Gurevich.
\newblock Abstract {S}tate {M}achines capture parallel algorithms: Correction
  and extension.
\newblock {\em ACM Trans. Comp. Logic}, 9(3), 2008.

\bibitem{boerger:fac2003}
E.~B{\"o}rger.
\newblock The {ASM} refinement method.
\newblock {\em Formal Aspects of Computing}, 15:237--257, 2003.

\bibitem{boerger:ai2016}
E.~B\"orger and K.-D. Schewe.
\newblock Concurrent {Abstract State Machines}.
\newblock {\em Acta Informatica}, 53(5):469--492, 2016.

\bibitem{boerger:jucs2017}
E.~B{\"{o}}rger and K.-D. Schewe.
\newblock Communication in {Abstract State Machines}.
\newblock {\em J. Univ. Comp. Sci.}, 23(2):129--145, 2017.

\bibitem{boerger:2003}
E.~B\"orger and R.~St\"{a}rk.
\newblock {\em {Abstract State Machines}}.
\newblock Springer-Verlag, Berlin Heidelberg New York, 2003.

\bibitem{dershowitz:igpl2016}
N.~Dershowitz and E.~{Falkovich-Derzhavetz}.
\newblock On the parallel computation thesis.
\newblock {\em Logic Journal of the {IGPL}}, 24(3):346--374, 2016.

\bibitem{ferrarotti:psi2017}
F.~Ferrarotti, K.-D. Schewe, and L.~Tec.
\newblock A behavioural theory for reflective sequential algorithms.
\newblock In A.~K. Petrenko and A.~Voronkov, editors, {\em Perspectives of
  System Informatics -- 11th International Andrei P. Ershov Informatics
  Conference (PSI 2017)}, volume 10742 of {\em LNCS}, pages 117--131. Springer,
  2017.

\bibitem{ferrarotti:tcs2016}
F.~Ferrarotti, K.-D. Schewe, L.~Tec, and Q.~Wang.
\newblock A new thesis concerning synchronised parallel computing -- simplified
  parallel {ASM} thesis.
\newblock {\em Theor. Comp. Sci.}, 649:25--53, 2016.

\bibitem{ferrarotti:igpl2017}
F.~Ferrarotti, K.-D. Schewe, L.~Tec, and Q.~Wang.
\newblock A complete logic for {Database Abstract State Machines}.
\newblock {\em The Logic Journal of the {IGPL}}, 25(5):700--740, 2017.

\bibitem{ferrarotti:amai2018}
F.~Ferrarotti, K.-D. Schewe, L.~Tec, and Q.~Wang.
\newblock A unifying logic for non-deterministic, parallel and concurrent
  {Abstract State Machines}.
\newblock {\em Ann. Math. Artif. Intell.}, 83(3-4):321--349, 2018.

\bibitem{gurevich:lipari1995}
Y.~Gurevich.
\newblock {Evolving algebras 1993: Lipari Guide}.
\newblock In {\em Specification and Validation Methods}, pages 9--36. Oxford
  University Press, 1995.

\bibitem{gurevich:tocl2000}
Y.~Gurevich.
\newblock Sequential abstract-state machines capture sequential algorithms.
\newblock {\em ACM Trans. Comp. Logic}, 1(1):77--111, 2000.

\bibitem{lamport:tc1979}
L.~Lamport.
\newblock How to make a multiprocessor computer that correctly executes
  multiprocess programs.
\newblock {\em IEEE Trans. Computers}, 28(9):690--691, 1979.

\bibitem{lynch:1996}
N.~Lynch.
\newblock {\em Distributed Algorithms}.
\newblock Morgan Kaufmann, 1996.

\bibitem{riccobene:abz2014}
E.~Riccobene and P.~Scandurra.
\newblock Towards {ASM}-based formal specification of self-adaptive systems.
\newblock In Y.~A. Ameur and K.-D. Schewe, editors, {\em Abstract State
  Machines, Alloy, B, TLA, VDM, and {Z} - 4th International Conference ({ABZ}
  2014)}, volume 8477 of {\em Lecture Notes in Computer Science}, pages
  204--209. Springer, 2014.

\bibitem{schewe:acsw2017}
K.-D. Schewe, F.~Ferrarotti, L.~Tec, Q.~Wang, and W.~An.
\newblock Evolving concurrent systems: behavioural theory and logic.
\newblock In {\em Proceedings of the Australasian Computer Science Week
  Multiconference, ({ACSW} 2017)}, pages 77:1--77:10. {ACM}, 2017.

\bibitem{schewe:ac2010}
K.-D. Schewe and Q.~Wang.
\newblock A customised {ASM} thesis for database transformations.
\newblock {\em Acta Cybernetica}, 19(4):765--805, 2010.

\bibitem{schewe:jucs2010}
K.-D. Schewe and Q.~Wang.
\newblock {XML} database transformations.
\newblock {\em J. {UCS}}, 16(20):3043--3072, 2010.

\bibitem{schewe:ejc2011}
K.-D. Schewe and Q.~Wang.
\newblock Partial updates in complex-value databases.
\newblock In A.~Heimb\"{u}rger et~al., editors, {\em Information and Knowledge
  Bases XXII}, volume 225 of {\em Frontiers in Artificial Intelligence and
  Applications}, pages 37--56. IOS Press, 2011.

\bibitem{schewe:abz2012}
K.-D. Schewe and Q.~Wang.
\newblock A simplified parallel {ASM} thesis.
\newblock In J.~Derrick et~al., editors, {\em Abstract State Machines, Alloy,
  B, VDM, and Z (ABZ 2012)}, volume 7316 of {\em LNCS}, pages 341--344.
  Springer, 2012.

\bibitem{smith:popl1984}
B.~C. Smith.
\newblock Reflection and semantics in {LISP}.
\newblock In {\em Proceedings of the 11th ACM SIGACT-SIGPLAN Symposium on
  Principles of Programming Languages}, POPL '84, pages 23--35. ACM, 1984.

\bibitem{stemple:2000}
D.~Stemple, L.~Fegaras, R.~Stanton, T.~Sheard, P.~Philbrow, R.~Cooper,
  M.~Atkinson, R.~Morrison, G.~Kirby, R.~Connor, and S.~Alagic.
\newblock Type-safe linguistic reflection: A generator technology.
\newblock In M.~Atkinson and R.~Welland, editors, {\em Fully Integrated Data
  Environments}, Esprit Basic Research Series, pages 158--188. Springer Berlin
  Heidelberg, 2000.

\bibitem{bussche:jcss1996}
J.~{Van den Bussche}, D.~{Van Gucht}, and G.~Vossen.
\newblock Reflective programming in the relational algebra.
\newblock {\em J. Comput. Syst. Sci.}, 52(3):537--549, 1996.

\end{thebibliography}

\end{document}